\documentclass[11pt]{article}

\usepackage{fullpage}
\usepackage{times}
\usepackage{amsmath,amsfonts,amssymb,amsthm}
\usepackage{enumerate}
\usepackage{xcolor}
\usepackage{xspace}
\usepackage{tikz}

\newtheorem{definition}{Definition}
\newtheorem{theorem}{Theorem}
\newtheorem{lemma}{Lemma}
\newtheorem{corollary}{Corollary}
\newtheorem{proposition}{Proposition}
\newtheorem{claim}{Claim}
\newtheorem*{remark}{Remark}
\newcommand{\E}{\mathop{\mathbb{E}}}

\newcommand{\cD}{\mathcal{D}}
\newcommand{\URsup}{$\textbf{UR}^{\subset}$}
\newcommand{\URsub}{\textbf{UR}^{\subset}}
\newcommand{\UR}{$\textbf{UR}$}
\newcommand{\ENC}{\mathsf{Enc}}
\newcommand{\DKL}{D_{\mathrm{KL}}}
\newcommand{\DKLver}[2]{D_{\mathrm{KL}}\left(\ooalign{$\genfrac{}{}{1.6pt}0{#1}{#2}$\cr$\color{white}\genfrac{}{}{.8pt}0{\phantom{#1}}{\phantom{#2}}$}\right)}
\newcommand{\sk}{\mathsf{sk}}
\newcommand{\ur}{\mathsf{ur}}

\newcommand{\supp}{\mathop{support}}
\newcommand{\suppfind}[1]{support-finding$_{{#1}}$}
\newcommand{\R}{\mathbb{R}}
\newcommand{\g}{\mathbf{g}}
\newcommand{\h}{\mathbf{h}}
\newcommand{\icost}{\mathcal{I}}

\let\E\relax
\DeclareMathOperator*{\E}{\mathbb{E}}
\let\Pr\relax
\DeclareMathOperator*{\Pr}{\mathbb{P}}

\begin{document}
\setcounter{page}{0}

\author{
Jelani Nelson\thanks{Harvard University. \texttt{minilek@seas.harvard.edu}. Supported by NSF grant IIS-1447471 and CAREER award CCF-1350670, ONR Young Investigator award N00014-15-1-2388 and DORECG award N00014-17-1-2127, an Alfred P.\ Sloan Research Fellowship, and a Google Faculty Research Award.}
  \and Huacheng Yu\thanks{Harvard University. \texttt{yuhch123@gmail.com}. Supported by ONR grant N00014-15-1-2388, a Simons Investigator Award and NSF Award CCF 1715187.}
}

\title{Optimal Lower Bounds for Distributed and Streaming Spanning Forest Computation}

\maketitle

\thispagestyle{empty}

\begin{abstract}
We show optimal lower bounds for spanning forest computation in two different models:
\begin{itemize}
\item One wants a data structure for fully dynamic spanning forest in which updates can insert or delete edges amongst a base set of $n$ vertices. The sole allowed query asks for a spanning forest, which the data structure should successfully answer with some given (potentially small) constant probability $\epsilon>0$. We prove that any such data structure must use $\Omega(n\log^3 n)$ bits of memory.
\item There is a referee and $n$ vertices in a network sharing public randomness, and each vertex knows only its neighborhood; the referee receives no input. The vertices each send a message to the referee who then computes a spanning forest of the graph with constant probability $\epsilon>0$. We prove the average message length must be $\Omega(\log^3 n)$ bits.
\end{itemize}

Both our lower bounds are optimal, with matching upper bounds provided by the AGM sketch \cite{AhnGM12} (which even succeeds with probability  $1 - 1/\mathrm{poly}(n)$). Furthermore, for the first setting we show optimal lower bounds even for low failure probability $\delta$, as long as $\delta > 2^{-n^{1-\epsilon}}$.
\end{abstract}

\newpage

\section{Introduction}\label{sec:intro}
Consider the incremental spanning forest data structural problem: edges are inserted into an initially empty undirected graph $G$ on $n$ vertices, and the data structure must output a spanning forest of $G$ when queried. The optimal space complexity to solve this problem is fairly easy to understand. For the upper bound, one can in memory maintain the list of edges in some spanning forest $F$ of $G$, using $O(|F|\log n) = O(n\log n)$ bits of memory. To process the insertion of some edge $e$, if its two endpoints are in different trees of $F$ then we insert $e$ into $F$; else we ignore $e$. The proof that this data structure uses asymptotically optimal space is straightforward. Consider that the following map from trees on $n$ labeled vertices must be an injection: fix a correct data structure $D$ for this problem, then for a tree $T$ feed all its edges one by one to $D$ then map $T$ to $D$'s memory configuration. This map must be an injection since $D.\mathbf{query}()$ will be different for different $T$ (the query result must be $T$ itself!). If $D$ uses $S$ bits of memory then it has at most $2^S$ distinct possible memory configurations. Since the set of all trees has size $n^{n-2}$ by Cayley's formula, we must thus have $S \ge (n-2)\log n$. A similar argument shows the same asymptotic lower bound even for Monte Carlo data structures which must only succeed with constant probability: by an averaging argument, there must exist a particular random seed that causes $D$ to succeed on a constant fraction of all spanning trees. Fixing that seed then yields an injection from a set of size $\Omega(n^{n-2})$ to $\{0,\ldots,2^S - 1\}$, yielding a similar lower bound.

What though is the optimal space complexity to solve the fully dynamic case, when the data structure must support not only edge insertions, but also deletions? The algorithm in the previous paragraph fails to generalize to this case, since if an edge $e$ in the spanning forest $F$ being maintained is deleted, without remembering the entire graph it is not clear how to identify an edge to replace $e$ in $F$. Surprisingly though, it was shown in \cite{AhnGM12} (see also \cite{KapronKM13}) that there exists a randomized Monte Carlo data structure, the ``AGM sketch'', solving the fully dynamic case using $O(n\log^3 n)$ bits of memory with failure probability $1/\mathrm{poly}(n)$. The sketch can also be slightly re-parameterized to achieve failure probability $1-\delta$ for any $\delta\in(0,1)$ while using $O(n\log(n/\delta)\log^2 n)$ bits of memory (see Appendix~\ref{sec:agm-small-fail}). Our first main result is a matching lower bound for nearly the full range of $\delta\in(0,1)$ of interest. Previously, no other lower bound was known beyond the simple $\Omega(n\log n)$ one already mentioned for the incremental case.

\paragraph{Our Contribution I.} We show that for any $2^{-n^{1-\epsilon}}<\delta<1-\epsilon$ for any fixed positive constant $\epsilon>0$, any Monte Carlo data structure for fully dynamic spanning forest with failure probability $\delta$ must use $\Omega(n\log(n/\delta)\log^2 n)$ bits of memory. Note that this lower bound cannot possibly hold for $\delta < 2^{-n}$, since there is a trivial solution using $\binom n2$ bits of memory achieving $\delta=0$ (namely to remember exactly which edges exist in $G$), and thus our lower bound holds for nearly the full range of $\delta$ of interest.

\bigskip

One bonus feature of the AGM sketch is that it can operate in a certain distributed sketching model as well. In this model, $n$ vertices in an undirected graph $G$ share public randomness together with an $(n+1)$st party we will refer to as the ``referee''. Any given vertex $u$ knows only the vertices in its own neighborhood, and based only on that information and the public random string must decide on a message $M_u$ to send the referee. The referee then, from these $n$ messages and the public random string, must output a spanning forest of $G$ with success probability $1-\delta$. The AGM sketch implies a protocol for this model in which the maximum length of any $M_u$ is $O(\log^3 n)$ bits, while the failure probability is $1/\mathrm{poly}(n)$. As above, this scheme can be very slightly modified to achieve failure probability $\delta$ with maximum message length $O(\log(n/\delta)\log^2 n)$ for any $\delta\in(0,1)$ (see Appendix~\ref{sec:agm-small-fail}).

\paragraph{Our Contribution II.} We show that in the distributed sketching model mentioned above, even success probability $\epsilon$ for arbitrarily small constant $\epsilon$ requires even the {\it average} message length to be $\Omega(\log^3 n)$ bits.

\bigskip

We leave it as an open problem to extend our distributed sketching lower bound to the low failure probability regime. We conjecture a lower bound of $\Omega(\log(n/\delta)\log^2 n)$ bits for any $\delta > 2^{-n^{1-\epsilon}}$.

Despite our introduction of the two considered problems in the above order, we show our results in the opposite order since we feel that our distributed sketching lower bound is easier to digest. In Section~\ref{sec_dist} we show our distributed sketching lower bound, and in Section~\ref{sec_str} we show our data structure lower bound. Before delving into the proof details, we first provide an overview of our approach in Section~\ref{sec:overview}.

\section{Proof Overview}\label{sec:overview}

The starting point for both our lower bound proofs is the randomized one-way communication complexity of {\it universal relation} (\UR) in the public coin model, for which the first optimal lower bound was given in \cite{KapralovNPWWY17}. In this problem Alice and Bob receive sets $S,T\subseteq[U]$, respectively, with the promise that $S\neq T$. Bob then, after receiving a single message from Alice, must output some $i$ in the symmetric difference $S\triangle T$. We will specifically be focused on the special case $\URsub$ in which we are promised that $T\subsetneq S$. In \cite{KapralovNPWWY17} it is shown that for any $\delta\in(0,1)$ bounded away from $1$, the one-way randomized communication complexity of this problem in the public coin model satisfies $R^{pub,\rightarrow}_\delta(\URsub) = \Theta(\min\{U, \log(1/\delta)\log^2(U/\log(1/\delta))\})$. Note that by Yao's minimax principle, this implies the existence of a ``hard distribution'' $\cD_\ur$ over $(S,T)$ pairs such that the distributional complexity under $\cD_\ur$ satisfies $D^{\cD_\ur,\rightarrow}_\delta = \Theta(\min\{U, \log(1/\delta)\log^2(U/\log(1/\delta))\})$.

\subsection{Distributed sketching lower bound}

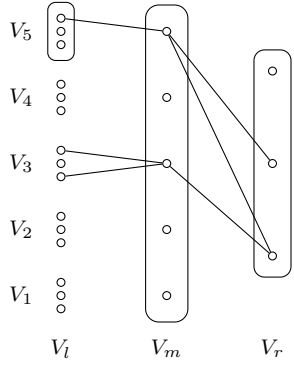
\begin{figure}
\centering
\begin{tikzpicture}[vtx/.style={draw, circle, inner sep=0, minimum size=3pt}]
	\foreach \i in {0,1,2,3,4}
		\node[vtx] (A\i) at (0pt, \i*25 pt) {};
	\foreach \i in {0,1,2,3,4}
		\foreach \j in {0,1,2}
			\node[vtx] (B\i\j) at (-40pt, \i*25+\j*5-5 pt) {};
	\foreach \i in {0,1,2}
		\node[vtx] (C\i) at (40pt, \i*35+15 pt) {};
	\draw [rounded corners=5pt] (-8pt, -10pt) rectangle (8pt, 110pt);
	\draw [rounded corners=3pt] (-45pt, 89pt) rectangle (-35pt, 111pt);
	\draw [rounded corners=4pt] (33pt, 7pt) rectangle (47pt, 93pt);

	\node at (-40pt, -20pt) {\scriptsize $V_l$};
	\node at (0, -20pt) {\scriptsize $V_m$};
	\node at (40pt, -20pt) {\scriptsize $V_r$};
	\foreach \i in {1,2,3,4,5}
		\node at (-55pt, \i*25-25 pt) {\scriptsize $V_\i$};

	\draw (B42) -- (A4) -- (C1);
	\draw (A4) -- (C0);

	\draw (B22) -- (A2) -- (C0);
	\draw (B20) -- (A2);
\end{tikzpicture}
\caption{Hard instances for computing spanning forest.}\label{fig:hard}
\end{figure}

Our lower bound in the distributed sketching model comes from a series of two reductions. Assume there is a protocol $P$ on $n$-vertex graphs with expected average message length $L = o(\log^3 n)$ (for the sake of contradiction) and failure probability $1/3$, say (our argument extends even to failure probability $1-\epsilon$ for constant $\epsilon$). We show this implies that for any distribution $\cD_\sk$ over $n^{4/5}$-vertex graphs there is a protocol with failure probability at most $O(1/\mathrm{poly}(n))$ and expected message length at most $O(L)$. We then use this to show that for any distribution $\cD_\ur$ for $\URsub$ over a universe of size $U = n^{1/5}$, there exists a protocol with failure probability $O(\sqrt{L}/\mathrm{poly}(n)) = 1/\mathrm{poly}(n)$ and expected average message length $O(L)$, a contradiction, since it violates the lower bound of \cite{KapralovNPWWY17}.

We sketch the reduction from $\URsub$ to the graph sketching problem via Figure~\ref{fig:hard}. Suppose Alice and Bob are trying to solve an instance of $\URsub$, where they hold $S,T\subset[n^{1/5}]$. We set $|V_m| = \frac 12 n^{3/5}$, $|V_l| = |V_m|\cdot |V_r|$, and $|V_r|$ as well as the size of each block in $V_l$ equals $n^{1/5}$. Thus overall there are at most $|V| = n^{4/5}$ vertices. Both Alice and Bob will agree that the vertices in $V_m$ are named $v_1,v_2,\ldots,v_{|V_m|}$. The main idea is that $T$ will correspond to the neighbors of $v_i$ in the $i$th block of $V_l$ (we call this $i$th block ``$V_i$''), and any neighbors in $V_r$ correspond to elements of $S\setminus T$. Since $V_i$ may only connect to $v_i$, in any spanning forest of $G$ the only way that $V_i\cup\{v_i\}$ connects to the rest of the graph is from an edge between $v_i$ and $V_r$, i.e., finding a spanning forest allows one to recover one element in $S\setminus T$. To find a spanning forest, Alice and Bob would like to simulate the distributed sketching protocol on $G$. However, $S\setminus T$ is not known to either of the players, which the messages from $V_r$ depend on, hence Alice and Bob might not be able to simulate the protocol perfectly. We resolve this issue by exploiting the fact that $L\cdot |V_r| = o(|V_m|)$, and thus all the messages from $V_r$ combined only reveal $o(1)$ bits of information about the neighborhood of a random $v_i\in V_m$ and are thus unimportant for Alice and Bob to simulate perfectly.

The remainder of the sketch of the reduction is then as follows. Alice and Bob also use public randomness to pick a random injection $\beta:[n^{1/5}]\rightarrow[n^{4/5}]\setminus V_m$, and also to pick a random $i\in [|V_m|]$. They then attempt to embed their $\URsub$ instance in the neighborhood of $v_i$. Alice sends Bob the message $\sk(v_i)$ to Bob, as if $v_i$ had neighborhood $\beta(S)$. Bob then picks vertex names for $V_l, V_r$ randomly in $[n^{4/5}]\backslash V_m$ conditioned on $\beta(T)\subset V_i$ and $\beta([n^{1/5}]\backslash T)\subset V_r$. Then for all $j\neq i$, Bob samples random $S_j,T_j$ from $\cD_\ur$ and connects $v_j$ to $|T_j|$ random vertices in $V_j$ and $|S_j\backslash T_j|$ random vertices in $V_r$. Bob then computes all the sketches of every vertex other than $v_i$ then simulates the referee to output the $u\in V_r$ which maximizes the probability that $(v_i, u)$ is an edge, conditioned on $V_l,V_m,V_r,\sk(V_l),\sk(V_m)$.

\subsection{Data structure space lower bound}
Our data structure lower bound comes from a variant of a direct product theorem of \cite{BravermanRWY13} (we will explain the relevance soon). Their work had two main theorems: the first states that for any boolean function $f(x,y)$ and distribution $\mu$, if $C$ is such that the smallest achievable failure probability of any protocol for $f$ with communication cost $C$ on distribution $\mu$ is $\gamma$, then any protocol for $f^n$ (the $n$-fold product of $f$) on distribution $\mu^n$ with communication at most $T = \tilde{O}(\gamma^{5/2}C\sqrt n)$ must have success probability at most $\exp(-\Omega(\gamma^2 n))$. The second theorem is similar but only works for $\mu$ a product distribution, but with the benefit that the communication cost for $f^n$ need only be restricted to $T = \tilde{O}(\gamma^6 C n)$; this second theorem though does not apply, since one would want to apply this theorem with $\mu$ being a hard distribution $\cD_{\ur}$ for $\URsub$, which clearly cannot be a product distribution (Bob's input is promised to be a subset of Alice's, which means in a product distribution there must be some $D$ such that $T\subseteq D\subsetneq S$ always, in which case Alice can send one element of $S\backslash D$ using $O(\log U)$ bits and have zero error). In any case, even if $\cD_\ur$ were a product distribution, these theorems are too weak for our purposes. This is because the way in which one would {\it like} to apply such a direct product theorem is as follows. First, we would like to reduce $f^n$ to fully dynamic spanning forest for $f = \URsub$ (we give such a reduction in Section~\ref{sec_str_red}). Such a reduction yields that if a $T$-bit memory solution for fully dynamic spanning forest with success probability $1-\delta$ existed over a certain distribution over graphs, it would yield a one-way $T$-bit protocol for $f^n$ with success probability $1-\delta$ over $\mu^n$. Next, the natural next course of action would be to apply the {\it contrapositive} of such a direct product theorem: if a $T$-bit protocol for $f^n$ with success probability $\exp(-c \gamma^2 n) = 1-\delta$ exists over $\mu^n$, then there must exist a $C$-bit communication protocol for $f$ with failure probability $\gamma = O(\sqrt{\delta/n})$ over $\mu$. By the main result of \cite{KapralovNPWWY17} any such protocol must use $D^{\mu,\rightarrow}_{\sqrt{\delta/n}}(\URsub) = \Omega(\log(n/\delta)\log^2 n)$ bits of space (in our reduction $U = n$), so if the $C$ we obtained is less than this, then we would arrive at a contradiction, implying that our initial assumption that such a $T$-bit data structure for spanning forest exists must be false. Unfortunately the relationship between $C$ and $T$ in \cite{BravermanRWY13} is too weak to execute this strategy. In particular, we would like prove a lower bound for our $f^n$ of the form $D^{\mu^n,\rightarrow}_{\delta}(f^n) = \Omega(n\cdot D^{\mu,\rightarrow}_{\sqrt{\delta/n}}(f)) := n\cdot C$, where $D$ denotes distributional complexity. That is, we would like to obtain hardness results for $T = \Omega(n\cdot C)$, but the theorems of \cite{BravermanRWY13} only allow us to take $T$ much smaller; i.e.\ the first theorem requires $T = \tilde{O}(\gamma^{5/2}C\sqrt n)$ (and recall $\gamma = \Theta(\sqrt{\delta/n})$).

The main observation is that if one inspects the proof details in \cite{BravermanRWY13}, one discovers the following intermediate result (not stated explicitly as a lemma, but implicit in their proofs). We state now the restriction of this intermediate result to one-round, one-way protocols, which is what is relevant in our setting. Suppose for some boolean function $f(x,y)$ there exists a one-way protocol $P$ with failure probability $\delta$ and communication cost $T$ for $f^n$ on some $n$-fold product distribution $\mu^n$. Then there exists a distribution $\theta$ over triples and one-way protocol $P'$ for $f$ such that if $\pi$ is the distribution over $(X,Y,M)$, where $(X,Y)\sim\mu$ and $M$ Alice's message (which is a function of only $X$ and her private randomness), then
\begin{itemize}
\item the failure probability of $P'$ for inputs generated according to $\mu$ is $O(\sqrt{\delta/n})$, 
\item $\theta$ and $\pi$ are ``close'' (for some notion of closeness)
\item the {\it internal information cost} with respect to $\theta$, $I_\theta(M ; X | Y)$, is $O(T/n)$.
\end{itemize}

One point to note is that the $\theta$ distribution above is not guaranteed to correspond to a valid communication protocol, i.e.\ for $(X,Y,M)\sim \theta$, we are not promised that $M$ is a function of only $X$ and Alice's private randomness. In order to make use of the above direct product theorem, we then prove a distributional lower bound for $\URsub$ for some hard distribution $\cD_\ur$ which states that for any distribution $\theta$ as above, the internal information cost $I_\theta(M ; X | Y)$ must be at least the $\URsub$ lower bound mentioned above. Such a proof follows with minor differences from the proof in \cite{KapralovNPWWY17}; we provide all the details in Appendix~\ref{app_urlb}.

We remark that other works have provided related direct sum or direct product theorems, e.g.\ \cite{BarYossefJKS04,BarakBCR13,MolinaroWY13,Jain15,JPY16,BRWY13ICALP}. The work \cite{BarYossefJKS04} (see \cite[Theorem 1.5]{BarakBCR13} for a crisp statement) proves a direct sum theorem for computing $f$ on $n$ independent input drawn from some distribution $\mu$, under internal information cost. The downside of direct sum theorems, as studied in this work and \cite{BarakBCR13}, is that they only aim to show that the cost of computing $n$ copies of a function is at least $n$ times the cost of computing one copy with the same failure probability. In our case though, we would like to argue that computing $n$ copies requires {\it more} than $n$ times the cost, since we would like to say that computing $n$ copies of $\URsub$ with overall failure probability say, $1/3$, requires $n$ times the cost of computing a single copy with failure probability $O(1/\sqrt n)$, i.e.\ the cost multiplies by an $n\log n$ factor. Such theorems, which state that the success probability of low-cost protocols must go down quickly as $n$ increases, are known in the literature as {\it direct product} theorems. A direct product theorem similar to what we want in our current application was shown in \cite{MolinaroWY13}, but unfortunately the cost of computing $f^n$ with failure probability $\delta$ in that work is related to the cost of computing $f$ by a protocol that fails with probability $\delta/n$ (which is what we want) but that is {\it allowed to output `\textsf{Fail}'}, i.e.\ abort, with constant probability! Thus the main theorem in that work cannot be used to obtain a tight lower bound, since it is known that if one is allowed to abort with constant probability, then there is a $\URsub$ protocol that is actually a factor $\log n$ cheaper \cite{KapralovNPWWY17}.
The works \cite{Jain15,JPY16,BRWY13ICALP} are the most relevant, which proved direct product theorems with different trade-offs.
Similar to the situation of \cite{BravermanRWY13}, their direct product theorems may not be applied as a black-box to our lower bound.
However, by careful examinations of their proofs, it is possible to obtain a similar result to the one we derived from \cite{BravermanRWY13}.
Another point that one needs to pay attention to in the previous direct product theorems is that, they usually work in the regime where the error-per-instance $\epsilon$ is a constant, and prove that overall success probability must be exponentially small, whereas in our application, the overall failure probability is close to zero, and $\epsilon$ needs to be polynomially small.
A few steps in the previous arguments may lose $1/\epsilon$ factors, which is not crucial in their regime, but should be avoided in ours.

\section{Distributed Sketching Lower Bound}\label{sec_dist}

Given an undirected graph $G$ on $n$ vertices indexed by $[n]$.
Any given vertex only knows its own index and the set of indices of its neighbors, as well as a shared random string.
Then each vertex $v$ sends a message (a sketch $\sk(v)$) to a referee, who based on the sketches and the random string must output a spanning forest of $G$ with probability $1-\delta$. 
The task is to minimize the average size of the $n$ sketches.

In this section, we prove Theorem~\ref{thm_dist}, a sketch size lower bound for computing spanning forest in the distributed setting.

\begin{theorem}\label{thm_dist}
Any randomized distributed sketching protocol for computing spanning forest with \emph{success} probability $\epsilon$ must have expected average sketch size $\frac{1}{n}\E(\sum_v|\sk(v)|)\geq\Omega(\log^3 n)$, for any constant $\epsilon>0$.
\end{theorem}

The first observation is that since each node only sees its neighborhood, every message depends on a local structure of the graph.
If we partition the graph into, say $n^{1/5}$, components with $n^{4/5}$ nodes each, and put an independent instance in each component, then messages from each component are independent, and hence the referee has to compute a spanning forest for each instance with an overall success probability $\epsilon$, i.e., failure probability per instance is at most $O(n^{-1/5})$.
That is, it suffices to study the problem on slightly smaller graphs with a much lower error probability.

Next, we make a reduction from the communication problem \URsup{}.
In \URsup{}, Alice gets a set $S\subseteq [U]$, Bob gets a proper subset $T\subset S$.
Alice sends one single message $M$ to Bob.
The goal of the communication problem is to find one element $x\in S\setminus T$.
The one-way communication complexity with shared randomness is well understood~\cite{KapralovNPWWY17}.
\begin{theorem}[\cite{KapralovNPWWY17}]\label{thm_ur}
The randomized one-way communication complexity of \URsup{} with error probability $\delta$ in the public coin model is $\Theta(\min\{U, \log\frac{1}{\delta}\cdot\log^2 \frac{U}{\log 1/\delta}\})$.
\end{theorem}

To make the reduction, consider a vertex $v$ in the graph, $v$ sees neighborhood $N(v)$, and sends $\sk(v)$ to the referee.
Suppose that there is a subset $T$ of $N(v)$ such that for every vertex $u\in T$, $v$ is its only neighbor.
In this case, the only way that $v$ and $T$ connect to the rest of the graph is to go through an edge between $v$ and $N(v)\setminus T$, which the referee has to find and add to the spanning forest.
We may view $N(v)$ as a set $S$, vertex $v$ must commit the message $\sk(v)$ based only on $S$.
Then $T$ is revealed to the referee, who has to find an element in $S\setminus T$.
If the referee finds this element using only $\sk(v)$ (not the other sketches), then by Theorem~\ref{thm_ur}, the $|\sk(v)|$ must be at least $\Omega(\log^3 n)$.
In the proof, we will construct graphs such that for a (small) subset of vertices, the other sketches ``do not help much'' in finding their neighbors.
This would prove that the average sketch size of this subset of vertices must be at least $\Omega(\log^3 n)$.

Finally, to extend the lower bound to average size of all sketches, we further construct graphs where the neighborhood of each vertex looks like the neighborhood of a random vertex from this small subset.
In the final hard distribution, we put such a graph with constant probability, and a random instance from the last paragraph with constant probability.
Then prove that if the algorithm succeeds with high probability, its average sketch size must be large.

\begin{proof}[Proof of Theorem~\ref{thm_dist}]
Suppose there is a protocol $A$ for $n$-node graphs with error probability at most $1-\epsilon$ and expected average message length $\frac{1}{n}\sum_v\E|\sk(v)|=L$, then we have the following.

\begin{proposition}\label{prop_one}
For \emph{any} input distribution $\cD_{\sk}$ over $n^{4/5}$-node graphs, there is a deterministic protocol $A'_{\cD_{\sk}}$ with error probability $O(n^{-1/5})$ and expected average message length at most $O(L)$.
\end{proposition}
The main idea is to construct $n^{1/5}$ independent and disconnected copies of $n^{4/5}$-node instances, then simulate protocol $A$ on this whole $n$-node graph.
Then we show that since each message only depends on the neighborhood, the $n^{1/5}$ copies could only be solve independently.

Consider the input distribution on $n$-node graphs by independently sampling $n^{4/5}$-node graphs from $\cD_{\sk}$ on vertex sets $[n^{4/5}]$, $[n^{4/5}]+n^{4/5}$, $[n^{4/5}]+2n^{4/5}$, \ldots, and $[n^{4/5}]+n-n^{4/5}$.
Denote the resulting $n^{1/5}$ graphs by $G_1,\ldots,G_{n^{1/5}}$.
Denote by $G$ the union of the $n^{1/5}$ graphs. 
Protocol $A$ produces a spanning forest $F$ of $G$ with probability $\epsilon$.

Let us analyze $A$ on $G$, and let $R_A$ be the random bits used by $A$.
First, we may assume without loss of generality that $A$ is deterministic.
This is because by Markov's inequality, we have
\[
	\Pr_{R_A}\left(\frac{1}{n}\cdot\E_{G}|\sk(v)|>2L/\epsilon\right)<\frac{\epsilon}{2}
\]
and
\[
	\Pr_{R_A}\left(\Pr_{G}[A\textrm{ is wrong}]>1-\epsilon/2\right)<1-\frac{\epsilon}{2}.
\]
Thus, we may fix $R_A$ and hardwire it to the protocol such that the overall error probability (over a random $G$) is still at most $\frac{\epsilon}{2}$ and the expected average message length is at most $O(L)$.

Note that the message sent from a vertex in $G_i$ depends only on graph $G_i$ (in fact only its neighbors), and all $n^{1/5}$ graphs are independent a priori.
Therefore, after the referee sees all $n$ messages, conditioned on these messages, the $n^{1/5}$ input graphs are still independent. 
For any forest $F=F_1\cup\cdots\cup F_{n^{1/5}}$, where $F_i$ is a forest on vertices $[n^{4/5}]+(i-1)n^{4/5}$, the probability that $F$ is a spanning forest of $G$ is equal to the product of the probabilities that $F_i$ is a spanning forest of $G_i$.
Hence, to maximize the probability that the output is a spanning forest of $G$, we may further assume that $A$ outputs for each $i$, a forest $F_i$ on vertices $[n^{4/5}]+(i-1)n^{4/5}$ that maximizes the probability that $F_i$ is a spanning forest of $G_i$ conditioned on the messages.
Thus, all $n^{1/5}$ outputs $F_i$ also become independent, and overall success probability is equal to the product of success probabilities of all $n^{1/5}$ instances.
In particular, there exists an $i$ such that the probability that the output $F_i$ is a spanning forest of $G_i$ is at least $(1-\epsilon/2)^{n^{-1/5}}\geq 1-O(n^{-1/5})$.

To solve an $n^{4/5}$-node instance sampled from $\cD_{\sk}$, it suffices to embed the graph into $G_i$ of $G$.
Each vertex sends a message to the referee pretending themselves are nodes in $G_i$ using the above fixed random bits $R_A$, and the referee outputs an $F_i$ that is a spanning forest of $G_i$ with the highest probability conditioned on the messages.
Based on the above argument, the error probability is at most $O(n^{-1/5})$, and the expected average message length is at most $O(L)$.

\bigskip

Suppose such protocol exists, we must have the following solution for \URsup.

\begin{proposition}\label{prop_two}
For \emph{any} input distribution $\cD_{\ur}$, there is a one-way communication protocol for \URsup{} over universe $[n^{1/5}]$ with error probability $O(L^{1/2}\cdot n^{-1/5})$ and communication cost $O(L)$.
\end{proposition}
We first derive a \URsup{} protocol from a spanning forest protocol with \emph{worst-case} message length $O(L)$ and error probability $O(n^{-1/5})$, then extend the result to expected average length.

\paragraph{Hard instance $\cD_\sk$.}
Recall the graph with $n^{4/5}$ vertices from Figure~\ref{fig:hard}, which will be our spanning forest hard instance:
\begin{itemize}
	\item The vertex set $V$ is partitioned into four groups $V_l, V_m, V_r$ and $V_o$, all vertices in $V_o$ are isolated and hence can be ignored;
	\item $|V_l|=\frac12 n^{4/5}, |V_m|=\frac12 n^{3/5}$ and $|V_r|=n^{1/5}$;
	\item $V_l$ is further partitioned into $\frac12 n^{3/5}$ blocks $V_1,\ldots,V_{\frac12 n^{3/5}}$ such that each block $V_i$ contains $n^{1/5}$ vertices, and is associated with one vertex $v_i$ in $V_m$;
	\item The only possible edges in the graph are the ones between $V_m$ and $V_r$, and ones between block $V_i$ and the associated vertex $v_i$.
\end{itemize}

Let us consider the following distribution $\cD_{\sk}$ over such graphs:
\begin{enumerate}
	\item Sample a random $V_m$ of size $\frac12 n^{3/5}$, and sample disjoint $V_r, V_1, V_2,\ldots,V_{\frac12 n^{3/5}}$ such that each set has $n^{1/5}$ vertices;
	\item For each vertex $v_i\in V_m$, uniformly sample from $\cD_{\ur}$ a \URsup{} instance $(S_i, T_i)$ such that $S_i\supset T_i$;
	\item Connect each $v_i$ to uniformly random $|T_i|$ vertices in $V_i$, and to uniformly random $|S_i\setminus T_i|$ vertices in $V_r$.
\end{enumerate}

\paragraph{Reduction from \URsup.}
By Proposition~\ref{prop_one}, there exists a good spanning forest protocol $A'_{\cD_{\sk}}$ for the above distribution $\cD_{\sk}$.
Next, we are going to use this protocol to design an efficient one-way communication protocol $P$ for \URsup{} under $\cD_{\ur}$.
The main idea is to construct a graph $G$ as above and embed the \URsup{} instance to one of the neighborhoods of vertices $v_i\in V_m$, such that set $T$ corresponds to its neighbors in $V_i$ and $S\setminus T$ corresponds to its neighbors in $V_r$. 
Since $V_i$ may only connect to $v_i$, in any spanning forest of $G$, the only way that $V_i\cup\{v_i\}$ connects to the rest of the graph is from an edge between $v_i$ and $V_r$, i.e., finding a spanning forest allows one to recover one element in $S\setminus T$.
To find a spanning forest, we will simulate $A'_{\cD_{\sk}}$ on $G$.
However, $S\setminus T$ is not known to either of the players, which the messages from $V_r$ depend on, hence Alice and Bob might not be able to simulate the protocol perfectly.
We resolve this issue by exploiting the fact that $|V_r|\cdot L$ is much smaller than $V_m$, and thus the messages from $V_r$ do not reveal too much information about the neighborhood of a random $v_i\in V_m$.

More formally, first consider the following procedure to generate a random graph $G$ from two sets $S$ and $T$:
\begin{enumerate}
	\item sample a random $V_m$ of size $\frac12 n^{3/5}$, a uniformly random injection $\beta: [n^{1/5}] \rightarrow V\setminus V_m$, and a uniformly random vertex $v_i\in V_m$;
	\item sample uniformly random subsets $V_1,\ldots,V_{\frac12 n^{3/5}}$ and $V_r$ of size $n^{1/5}$ from the remaining vertices $V\setminus V_m$ \emph{conditioned on} $\beta(T)\subset V_i$ and $\beta([n^{1/5}]\setminus T)\subset V_r$;
	\item connect $v_i$ to all vertices in $\beta(S)$;
	\item for all other $v_j\in V_m$, sample $S_j$ and $T_j$ from $\cD_{\ur}$, and connect $v_j$ to $|T_j|$ random vertices in $V_j$ and $|S_j\setminus T_j|$ random vertices in $V_r$.
\end{enumerate}

Suppose $(S, T)$ is sampled from $\cD_{\ur}$, denote by $\mu$ the joint distribution of all random variables occurred in the above entire procedure.
To avoid lengthy subscripts, we denote the marginal distributions by $\mu[\cdot]$, e.g., denote by $\mu[S]$ the marginal distribution of $S$, and $\mu[S\mid G]$ the marginal of $S$ conditioned on $G$, etc.
Since $\beta$ is a uniformly random mapping, we have $\mu[G]=\cD_{\sk}$.
Moreover, if we find a spanning forest $F$ of $G$, in particular, a neighbor $u$ of $v_i$ in $V_r$, $\beta^{-1}(u)$ will be an element in $S\setminus T$.

We now give the protocol $P$ for \URsup{} (see Figure~\ref{fig:proP}), where the players attempt to sample a graph from $\mu[G\mid S, T]$, and exploit the fact that all sketches together would determine a spanning forest.

\begin{figure}[ht]
\begin{center}
\fbox{                                                                                                         
{\footnotesize                                                                                                 
\parbox{6.375in} {                                                                                             
\underline{\URsup{} Protocol $P$} (on input pair $(S, T)$ such that $T\subset S\subseteq[U]$ for $U=n^{1/5}$)\\[0.05in]
\textbf{initialization}
\vspace{-.05in}\begin{enumerate}
\addtolength{\itemsep}{-4pt}
\item sample a random $V_m$ of size $\frac12 n^{3/5}$, a uniformly random injection $\beta: [n^{1/5}] \rightarrow V\setminus V_m$, and a uniformly random vertex $v_i\in V_m$ using public random bits
\end{enumerate} 

\textbf{Alice}($S$)
\vspace{-.2in}\begin{enumerate}
\addtolength{\itemsep}{-4pt}
\setcounter{enumi}{1}
\medskip
\smallskip
\item simulate $A'_{\cD_{\sk}}$ as if she is vertex $v_i$ with neighborhood $\beta(S)$, and then send the sketch $\sk(v_i)$ to Bob
\end{enumerate}

\textbf{Bob}($T$)
\vspace{-.05in}
\begin{enumerate}
\addtolength{\itemsep}{-4pt}
\setcounter{enumi}{2}
\item sample uniformly random subsets $V_1,\ldots,V_{\frac12 n^{3/5}}$ and $V_r$ of size $n^{1/5}$ from the remaining vertices $V\setminus V_m$ \emph{conditioned on} $\beta(T)\subset V_i$ and $\beta([n^{1/5}]\setminus T)\subset V_r$
	\item for all other $v_j\in V_m$, sample $S_j$ and $T_j$ from $\cD_{\ur}$, and connect $v_j$ to $|T_j|$ random vertices in $V_j$ and $|S_j\setminus T_j|$ random vertices in $V_r$
	\item compute sketches $\sk(V_l)$ and $\sk(V_m)$\label{step_sk}
	\item find vertex $u\in V_r$, which maximizes $\mu((v_i, u)\textrm{ is an edge}\mid V_l, V_m, V_r, \sk(V_l), \sk(V_m))$, the probability that $(v_i, u)$ is an edge conditioned on the groups $V_l,V_m,V_r$ and sketches $\sk(V_l),\sk(V_m)$
	\item if $u\in \beta([n^{1/5}])$, output $\beta^{-1}(u)$, otherwise output an arbitrary element
\end{enumerate}
}}}
\caption{\URsup{} protocol using spanning forest protocol $A'_{\cD_{\sk}}$. 
In Step~\ref{step_sk}, Bob is able to compute all $\sk(V_l)$ and $\sk(V_m)$, because Bob knows the exact neighborhoods for all vertices in $V_l$ and $V_m\setminus \{v_i\}$, as well as the sketch $\sk(v_i)$ from Alice's message.}\label{fig:proP}
\end{center}
\end{figure}

\paragraph{Analyze $P$.}
The only message communicated is $\sk(v_i)$, which has $O(L)$ bits.
Next let us upper bound the error probability.

It is not hard to verify that in the protocol, the players sample $V_l, V_m, V_r, \sk(V_l), \sk(V_m), \beta$ from the right distribution $\mu[V_l, V_m, V_r, \sk(V_l), \sk(V_m), \beta\mid S, T]$.
To find an edge between $v_i$ and $V_r$, Bob computes the distribution $$\mu[E(v_i, V_r) \mid V_l, V_m, V_r, \sk(V_l), \sk(V_m)],$$ and returns the edge that occurs with highest probability, where $E(v_i, V_r)$ is the edges between $v_i$ and $V_r$.

On the other hand, given the sketches of all vertices, referee's algorithm outputs a spanning forest with probability $1-O(n^{-1/5})$.
In particular, it finds one edge between $v_i$ and $V_r$ (since in $\cD_{\sk}$, the only way to connect $V_i\cup \{v_i\}$ to the rest of the graph is through such an edge).
That is, in expectation, some edge $(v_i, u)$ has probability mass at least $1-O(n^{-1/5})$ in the distribution 
$$\mu[E(v_i, V_r)\mid V_l, V_m, V_r, \sk(V)].$$
In the following, we are going to show that the expected statistical distance between $\mu[E(v_i, V_r)\mid V_l, V_m, V_r, \sk(V)]$ and Bob's distribution $\mu[E(v_i, V_r) \mid V_l, V_m, V_r, \sk(V_l), \sk(V_r)]$ is small, which implies that the same edge $(v_i, u)$ would also appear in $\mu[E(v_i, V_r) \mid V_l, V_m, V_r, \sk(V_l), \sk(V_r)]$ with high probability.
By the definition of $\mu$, any edge $(v_i, u)$ for $u\in V_r$ corresponds to one element in $S\setminus T$.
Hence, Bob's error probability (over a random input pair) is small.

For simplicity of notation, we will omit $V_l, V_m, V_r$ in the conditions in the following. 
Fix $i$, we have
\begin{align*}
	&\ \E_{\mu}\left(\left|\mu[E(v_i, V_r)\mid \sk(V)]-\mu[E(v_i,V_r)\mid \sk(V_l), \sk(V_r)]\right|\right) \\
	\intertext{by Pinsker's inequality,}
	\leq&\  \E_{\mu}\left(\sqrt{\frac12 \DKLver{\mu[E(v_i, V_r)\mid \sk(V_l), \sk(V_m), \sk(V_r)]}{\mu[E(v_i,V_r)\mid \sk(V_l), \sk(V_m)]}}\right) \\
	\intertext{where $\DKL(q||p)=\DKLver{q}{p}$ is the Kullback--Leibler divergence from $p$ to $q$, and $|p - q|$ denotes statistical distance. Then by Jensen's inequality,}
	\leq&\  \sqrt{\frac12 \E_{\mu}\left(\DKLver{\mu[E(v_i, V_r)\mid \sk(V_l), \sk(V_m), \sk(V_r)]}{\mu[E(v_i,V_r)\mid \sk(V_l), \sk(V_m)]}\right)} \\
	=&\ \sqrt{\frac12 I(E(v_i, V_r); \sk(V_r)\mid \sk(V_l), \sk(V_m))}.
\end{align*}

By construction, conditioned on $\sk(V_l)$ and $\sk(V_m)$, the neighborhoods of vertices in $V_m$ are still independent.
Thus, by super-additivity of mutual information on independent variables,
\begin{align*}
	&\ \frac{1}{|V_m|}\sum_{i=1}^{|V_m|}I(E(v_i, V_r); \sk(V_r)\mid \sk(V_m), \sk(V_l)) \\
	\leq&\  \frac{1}{|V_m|} I(E(V_m, V_r); \sk(V_r)\mid \sk(V_m), \sk(V_l)) \\
	\leq&\  \frac{|\sk(V_r)|}{|V_m|} \\
	=&\ O(L\cdot n^{-2/5}).
\end{align*}

Finally, by another application of Jensen's inequality,
\begin{align*}
	&\ \E_{i, \mu}\left(\left|\mu[E(v_i, V_r)\mid \sk(V)]-\mu[E(v_i,V_r)\mid \sk(v_i), \sk(V_i)]\right|\right) \\
	\leq&\  \frac{1}{|V_m|}\sum_{i=1}^{|V_m|}\sqrt{\frac12 I(E(v_i, V_r); \sk(V_r)\mid \sk(V_m), \sk(V_l))} \\
	\leq&\  O(L^{1/2}\cdot n^{-1/5}).
\end{align*}
Thus, the error probability of $P$ is at most 
\begin{align*}
	\Pr(u\notin \beta(S\setminus T))&= \Pr((v_i, u)\textrm{ is not an edge}) \\
	&=\E(\mu((v_i, u)\textrm{ is not an edge}\mid V_l,V_m,V_r,\sk(V_l),\sk(V_m))) \\
	&\leq \E(\mu((v_i, u)\textrm{ is not an edge}\mid V_l,V_m,V_r,\sk(V)))+O(L^{1/2}\cdot n^{-1/5}) \\
	&\leq O(n^{-1/5})+O(L^{1/2}\cdot n^{-1/5}) \\
	&=O(L^{1/2}\cdot n^{-1/5}).
\end{align*}

\paragraph{Extend to expected average message length.}
To solve \URsup{} using protocols with only bounded expected average message length, we setup a new distribution over graphs where with $1/3$ probability, we sample a graph from the above hard distribution $\cD_{\sk}$; with $1/3$ probability, the neighborhoods of most vertices look as if they were in $V_m$; with $1/3$ probability, the neighborhoods of most vertices look as if they were in $V_r$.

More formally, first observe that each vertex in $V_m$ and each vertex in $V_r$ have the same degree distribution, denote the former by $\cD_{m}$ and the latter by $\cD_{r}$.
Moreover, conditioned on $v\in V_m$ and its degree, its neighborhood is uniformly random, and so is neighborhood of $v\in V_r$.
Finally, observe that the degree is at most $\frac{1}{2}n^{3/5}$.
Consider the following distribution $\cD'_{\sk}$:
\begin{enumerate}
	\item with probability $1/3$, sample $G$ from $\cD_{\sk}$;
	\item with probability $1/3$, randomly partition the vertices into two groups $U_1, U_2$ of $\frac{1}{2}n^{4/5}$ vertices each, for each vertex in $U_1$, sample its degree $d$ from $\cD_m$ and uniformly $d$ neighbors from $U_2$;
	\item with probability $1/3$, randomly partition the vertices into two groups $U_1, U_2$ of $\frac{1}{2}n^{4/5}$ vertices each, for each vertex in $U_1$, sample its degree $d$ from $\cD_r$ and uniformly $d$ neighbors from $U_2$.
\end{enumerate}

By Proposition~\ref{prop_one}, there is a protocol $A'$ with expected average message length $O(L)$ and error probability $O(n^{-1/5})$ on $\cD'_{\sk}$.
Let us analyze its performance on $\cD_{\sk}$.
From case 2 of $\cD'_{\sk}$, we conclude that the expected average message length of $V_m$ must be at most $O(L)$, as every vertex in $U_1$ has the same distribution of the neighborhood as a vertex in $V_m$.
Similarly, from case 3, the expected average message length of $V_r$ must be at most $O(L)$.
Finally, from case 1, the error probability must be at most $O(n^{-1/5})$.
Thus, by the previous argument, we obtain a \URsup{} protocol with \emph{expected} communication cost $O(L)$ and error probability $O(L^{1/2}\cdot n^{-1/5})$.

\medskip

By Theorem~\ref{thm_ur}, we have\footnote{Theorem~\ref{thm_ur} only states a lower bound for \emph{worst-case} communication cost of \URsup. One may verify that their proof works for expected communication cost as well. See also Lemma~\ref{lem:ur} for a \URsup{} lower bound in an even more general regime.}
\[
	L\geq \Omega\left(\log \frac{n^{1/5}}{L^{1/2}}\log^2 \frac{n^{1/5}}{\log\frac{n^{1/5}}{L^{1/2}}}\right)\geq \Omega\left(\log \frac{n^{1/5}}{L^{1/2}}\log^2 n\right).
\]
Thus, $L\geq \Omega(\log^3 n)$.
This proves the theorem.
\end{proof}

\section{Fully Dynamic Spanning Forest Data Structure}\label{sec_str}

In this section, we prove the following space lower bound for fully dynamic spanning forest data structures.

\begin{theorem}\label{thm:main-streaming}
Any Monte Carlo data structure for fully dynamic spanning forest with failure probability $\delta$ must use $\Omega(n \log \frac n{\delta}\log^2 n)$ bits of memory, as long as $\delta\in[2^{-n^{1-\epsilon}}, 1-\epsilon]$ for any given constant $\epsilon>0$.
\end{theorem}

We first observe that a good spanning forest data structure yields an efficient one-way communication protocol for $n$-fold \URsup{}.
In $n$-fold \URsup{}, Alice gets $n$ sets $S_1,\ldots,S_n\subseteq [U]$, Bob gets $n$ subsets $T_1,\ldots,T_n$ such that $T_i\subset S_i$ for all $i\in [n]$.
The goal is to find elements $x_i\in S_i\setminus T_i$ for all $i\in [n]$.
Then we prove a new direct product lemma for one-way communication based on the ideas from~\cite{BravermanRWY13}.
We show that a protocol for $n$-fold \URsup{} with cost $C$ and error $\delta$ gives us a new protocol for the original \URsup{} problem with ``cost'' $C/n$ and error $\sqrt{\delta/n}$, under a weaker notion of cost.
Then we generalize Theorem~\ref{thm_ur}, and show that same lower bound holds, which implies $C/n\geq \Omega(\log \frac{n}{\delta}\cdot \log^2 n)$.

In the following, we first show the reduction to $n$-fold \URsup{} in Section~\ref{sec_str_red}.
Then we prove the direct product lemma in Section~\ref{sec_str_dp}.
The proof of communication lower bound is deferred to Appendix~\ref{app_urlb}.

\subsection{Reduction to $n$-fold \URsup}\label{sec_str_red}
\begin{lemma}\label{lem_str_red}
If there is a fully dynamic data structure $A$ for spanning forest on a $2n$-node graph using $C$ bits of memory, and outputs a correct spanning forest with probability at least $1-\delta$, then there is a protocol for $n$-fold \URsup{} over $[n]$ using $C$ bits of communication with success probability $1-\delta$.
\end{lemma}
\begin{proof}
Consider a bipartite graph $G$ with $n$ nodes on each side.
For simplicity, we assume one side of the graph is indexed by the universe $[n]$, and the other side uses the same indices as the $n$ pairs of sets $(S_i, T_i)$.
Now we are going to simulate $A$ on a sequence of updates to $G$, and solve the communication problem.

Starting from the empty graph, Alice first simulates $A$.
For each pair $(x, i)$ such that $x\in S_i$, Alice inserts an edge between $x$ and $i$.
After all insertions, she sends the memory of $A$ to Bob, which takes $C$ bits of communication.
Then for each pair $(x, i)$ such that $x\in T_i$, Bob deletes the edge between $x$ and $i$.
After all deletions, Bob makes a query and obtains a spanning forest $F$ of $G$. 

For every non-isolated vertex, the spanning forest reveals one of its neighbors. 
In particular, any neighbor $x$ of a vertex $i$ (on the second side) must be in the set $S_i\setminus T_i$.
Therefore, it suffices to output for each $i$, an arbitrary neighbor of $i$ in $F$.
The overall success probability is at least $1-\delta$.
\end{proof}

Theorem~\ref{thm:main-streaming} is an immediate corollary of Lemma~\ref{lem_str_red} and the following lemma, which we prove in the remainder of the section.

\begin{lemma}\label{lem_str_lb}
Any one-way communication protocol for $n$-fold \URsup{} over universe $[n]$ with error probability $\delta$ must use $C\geq \Omega(n \log \frac{n}{\delta}\log^2 n)$ bits of communication, as long as $\delta\in[2^{-n^{1-\epsilon}}, 1-\epsilon]$ for any given constant $\epsilon>0$.
\end{lemma}

\subsection{Direct product lemma}\label{sec_str_dp}

\newcommand{\ighr}{{i,\g,\h,r}}

Consider one-way communication protocols with a fixed input distribution computing $f(X, Y)$.
A pair of inputs $(X, Y)$ is sampled from distribution $\cD$.
Two players Alice and Bob receive $X$ and $Y$ respectively. 
Alice sends one message $M$ to Bob.
Then Bob outputs a value $O$.

In the $n$-fold problem $f^n$, $n$ input pairs $(X_1, Y_1),\ldots, (X_n, Y_n)$ are sampled from $\cD$ independently. 
The goal is to compute all $f(X_1, Y_1),\ldots,f(X_n, Y_n)$.

In this subsection, we prove the following lemma which is implicitly proved in \cite{BravermanRWY13}.

\begin{definition}
Let $\cD$ be an input distribution for a two-player communication problem, a \emph{one-way $\theta$-protocol} consists of
\begin{itemize}
	\item distributions $M_{x,y}$ for every possible input $(x, y)$, and
	\item an output function $O=O(m, y)$ that takes a message $m$ and input $y$, and outputs a possible function value,
\end{itemize}
such that if $(X, Y)\sim \cD$ and $M\sim M_{X, Y}$, then $I(M ; Y\mid X)\leq \theta$.
\end{definition}
\begin{remark}
A \emph{one-way $0$-protocol} is a standard one-way communication protocol.
\end{remark}

\begin{lemma}\label{lem:dp}
Let $(X^{(n)}, Y^{(n)})\sim \cD^n$ be an input pair for an $n$-fold problem $f^n$, where $X^{(n)}=(X_1,\ldots,X_n)$ and $Y^{(n)}=(Y_1,\ldots,Y_n)$.
If there is a one-way protocol $\tau$ that takes input $(X^{(n)}, Y^{(n)})$, has \emph{communication} cost $C$ and computes $f^n$ with probability $p$, 
then there is an input distribution $\cD'$ for the one-fold problem $f$, and a one-way $O(\frac{1}{n}\log\frac1p)$-protocol $\pi$ such that
\begin{enumerate}
        \item $\DKL(\cD'[X]\|\cD[X]), \DKL(\cD'[Y]\|\cD[Y]) \le O\left(\log\frac 1p\right)$, where $\cD[X]$ denotes the marginal distribution of $X$ in $\cD$ (and similarly for the other terms).
        \item We have that both $\E_{x\sim \cD'[X]}\DKL(\cD'[Y|X=x]\|\cD[Y|X=x])$ and  $\E_{y\sim \cD'[Y]} \DKL(\cD'[X|Y=y]\|\cD[X|Y=y])$ are at most $O\left(\frac 1n\log\frac 1p\right)$, where $\cD[X|Y=y]$ denotes the marginal distribution of $X$ conditioned on $Y=y$ (and similarly for the other terms).
	\item under input distribution $\cD'$, $\pi$ computes $f$ with probability $1-O\left(\sqrt{\frac{1}{n}\log\frac{1}{p}}\right)$, and
	\item under input distribution $\cD'$, $\pi$ has ``internal information cost'' $I(X; M\mid Y)\leq O(C/n)$.
\end{enumerate}
\end{lemma}
\begin{proof}
	Let $i\in [n], \g, \h\subset [n]$ such that $i\notin \g\cup \h$, and $r$ be a possible assignment to $(X_\g, Y_\h)$.\footnote{$X_{\g}$ is $X^{(n)}$ restricted to coordinates $\g$ and $Y_{\h}$ is $Y^{(n)}$ restricted to coordinates $\h$.}
	To prove the lemma, we first define for every quadruple $(i, \g, \h, r)$, a $\theta$-protocol $\pi_{i, \g, \h, r}$ for the one-fold problem and an input distribution $\cD'_{i,\g,\h,r}$.
	Then we make a probabilistic argument showing that there is a carefully chosen distribution over $(i,\g,\h,r)$ such that in expectation $\pi_{i,\g,\h,r}$ and $\cD'_{i,\g,\h,r}$ have the desired properties.
	Finally, we finish the proof by applying Markov's inequality and union bound, and conclude that there is a quadruple that satisfies all requirements simultaneously.

	For the $n$-fold problem, under input distribution $\cD^n$ and protocol $\tau$, the inputs $(X^{(n)}, Y^{(n)})$, message $M$ and output $O$ form a joint distribution.
	Similar to the notations in Section~\ref{sec_dist},
	denote this distribution by $\mu$, and the marginal distribution of any subset of the random variables by $\mu[\cdot]$, e.g., the marginal distribution of $X$ and $M$ is denoted by $\mu[X, M]$, marginal distribution of $Y$ conditioned on event $V$ is denoted by $\mu[Y\mid V]$.
	Finally, denote by $W$ the event that output $O$ is correct.

	Now let us define $\pi_{i,\g,\h,r}$ and $\cD'_{i,\g,\h,r}$ (see Figure~\ref{fig:pi-ighr}).

	\begin{figure}[ht]
	\begin{center}
	\fbox{                                                                                                         
	{\footnotesize                                                                                                 
	\parbox{6.375in} {                                                                                             
	\underline{$\theta$-protocol $\pi_{i,\g,\h,r}$ for $f$}: \\[0.05in]
	\textbf{Distribution over messages} $M_{x,y}$:\\
	\vspace{-.05in}
	\hspace*{13pt}sample $m$ from $\mu[M\mid X_i=x,Y_i=y,(X_\g,Y_\h)=r,W]$\\

	\textbf{Output function} $O(m, y)$:
	\vspace{-.05in}
	\begin{enumerate}
	\addtolength{\itemsep}{-4pt}
	\item privately sample an output $o$ from $\mu[O\mid Y_i=y, (X_\g, Y_\h)=r,M=m,W]$ (with $n$ coordinates)
	\item output the $i$-th coordinate of $o$
	\end{enumerate}

	\bigskip

	\underline{Input distribution $\cD'_{i,\g,\h,r}$}: $(x, y)\sim\mu[X_i, Y_i\mid (X_\g,Y_\h)=r, W]$
	}}}
	\caption{$\pi_{i,\g,\h,r}$ on input pair $(x, y)$.}\label{fig:pi-ighr}
	\end{center}
	\end{figure}

	Denote the joint distribution of $X,Y,M,O$ under $\cD'_{i,\g,\h,r}$ and $\pi_{i,\g,\h,r}$ by $\nu_{i,\g,\h,r}$.
	Similarly, denote the marginal distribution of a subset of variables by $\nu_{i,\g,\h,r}[\cdot]$.
	Before we define the distribution of $(i,\g,\h,r)$, let us first analyze the information cost and success probability for each quadruple.

	\paragraph{``Distance'' from a protocol, and internal information cost of $\pi_\ighr$.} $\pi_\ighr$ is a $\theta$-protocol, where $\theta$ is
	\begin{equation}\label{eq_dist_prot}
		I_{\nu_{\ighr}}(M ; Y\mid X) = I_{\mu}(M ; Y_i\mid X_i, (X_\g, Y_\h)=r, W),
	\end{equation}
	because $\nu_\ighr[X,Y,M]$ is the same as $\mu[X_i,Y_i,M\mid (X_\g,Y_\h)=r,W]$.

	For the same reason, the internal information cost of $\pi_\ighr$ under input distribution $\cD'_\ighr$ is
	\begin{equation}\label{eq_info_cost}
		I_{\nu_{i,\g,\h,r}}\left(X;M\mid Y\right)=I_{\mu}\left(X_i;M\mid Y_i,(X_\g,Y_\h)=r, W\right).
	\end{equation}
	
	\paragraph{Error probability of $\pi_{i,\g,\h,r}$.}
	We first observe that $\mu\left[O\mid X_i=x,Y_i=y,(X_\g,Y_\h)=r,W\right]$ always has $f(x,y)$ in its $i$-th coordinate, since $W$ is the event that $\tau$ is correct.
	Thus, the statistical distance between $\mu\left[X_i,Y_i,M,O_i\mid (X_\g,Y_\h)=r, W\right]$ and $\nu_{i,\g,\h,r}[X,Y,M,O]$ will be an upper bound on the error probability.

	Same as above, the marginals of $(X,Y,M)$ are identical in the two distributions.
	To upper bound the statistical distance between $O$ in the two distributions conditioned on $(X, Y, M)$, for any $x, y$ and $m$, 
	\begin{align*}
		&\ \left|\mu[O_i\mid X_i=x,Y_i=y,M=m,(X_\g,Y_\h)=r, W]-\nu_\ighr[O\mid X=x,Y=y,M=m]\right| \\
		=&\ \left|\mu[O_i\mid X_i=x,Y_i=y,M=m,(X_\g,Y_\h)=r, W]-\mu[O_i\mid Y_i=y,M=m,(X_\g,Y_\h)=r, W]\right| \\
		\leq&\  O\left(\sqrt{\DKLver{\mu[O_i\mid X_i=x,Y_i=y,M=m,(X_\g,Y_\h)=r, W]}{\mu[O_i\mid Y_i=y,M=m,(X_\g,Y_\h)=r, W]}}\right).
	\end{align*}
	
	Thus, we have
	\begin{align}
          \nonumber	&\ \left|\mu[X_i,Y_i,M,O_i\mid (X_\g,Y_\h)=r, W]-\nu_\ighr[X,Y,M,O]\right| \\
	  \nonumber	=&\ \E[\left|\mu[O_i\mid X_i=x,Y_i=y,M=m,(X_\g,Y_\h)=r, W]-\nu_\ighr[O\mid X=x,Y=y,M=m]\right|] \\
	  \nonumber	\leq&\ \E_{(x,y,m)\sim \mu(X_i,Y_i,M\mid (X_\g,Y_\h)=r, W)}\left|O\left(\sqrt{\DKLver{\mu[O_i\mid X_i=x,Y_i=y,M=m,(X_\g,Y_\h)=r, W]}{\mu[O_i\mid Y_i=y,M=m,(X_\g,Y_\h)=r, W]}}\right)\right| \\
	  \intertext{which by Jensen's inequality and the fact that $I_p(X; Y)=\E_{y}\DKL(p[X\mid Y=y]\| p[X])$,}
		\le&\ O\left(\sqrt{I_{\mu}(X_i;O_i\mid Y_i,M,(X_\g,Y_\h)=r,W)}\right)\label{eq_err_prob}
	\end{align}
        which is an upper bound on the error probability.

	Now we are ready to define the distribution of $(i, \g, \h, r)$, and prove that all requirements are satisfied in expectation.
	\paragraph{Distribution of $(i, \g, \h, r)$.} There are two equivalent ways to generate the quadruple (see Figure~\ref{fig-dist-ighr} for one of them), which will be useful in different parts of the proof.

	Pick a uniformly random permutation $\kappa$ over $[n]$, pick two uniformly independent random numbers $s_g, s_h$ from the two halves respectively, i.e., $s_h\in[1, n/2]$, $s_g\in[n/2+1, n]$. Then
	\begin{itemize}
		\item set $i=\kappa(s_g)$, $\g=\kappa([1, s_g-1])$ and $\h=\kappa([s_h,n])\setminus \{i\}$; or
		\item set $i=\kappa(s_h)$, $\g=\kappa([1,s_g])\setminus \{i\}$ and $\h=\kappa([s_h+1,n])$.
	\end{itemize}

	\begin{figure}[ht]
	\begin{center}
	\begin{tikzpicture}
		\node at (-10pt, 0) {\scriptsize $\kappa:$};
		\foreach \i in {10,20,...,200}
			\node [draw, circle, inner sep=0, minimum size=5pt] at (\i pt, 0) {};
		\draw (105 pt, -5pt) -- (105 pt, 5pt);
		\draw [rounded corners=3pt, thick] (84pt, 8pt) rectangle (125pt, -5pt);
		\draw [rounded corners=3pt, thick] (5pt, 8pt) rectangle (76pt, -5pt);
		\draw [rounded corners=3pt, thick] (85pt, 5pt) rectangle (205pt, -8pt);
		\node at (80pt, 10pt) {\scriptsize $i$};
		\node at (40pt, 13pt) {\scriptsize $\g$};
		\node at (170pt, 10pt) {\scriptsize $\h$};
	\end{tikzpicture}
	\end{center}
	\caption{$(i, \g, \h)$ and $\kappa$ from the second distribution.}\label{fig-dist-ighr}
	\end{figure}
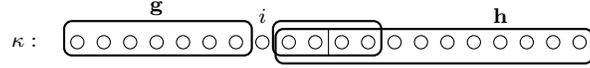

	The triple $(i,\g,\h)$ is identically distributed in the two distributions.
	Since $\kappa$ is a random permutation, then $i$ is a random element, $\g$ and $\h$ are two random sets of size $s_g-1$ and $n-s_h$ respectively, which has intersection size $s_g-s_h$.
	It is easy to verify that $i\notin \g\cup \h$ as required.
	Finally, we sample $r$ from $\mu[X_\g, Y_\h\mid W]$.

	\paragraph{The expected internal information cost is low.} To bound the internal information cost, we use the first view of the distribution of $(i,\g,\h)$.
	By Equation~\eqref{eq_info_cost}, the expected internal information cost is at most
	\begin{align*}
		&\ \E_\ighr I_{\nu_\ighr}\left(X;M\mid Y\right) \\
		=&\ \E_\ighr I_{\mu}\left(X_i;M\mid Y_i,(X_\g,Y_\h)=r, W\right) \\
		=&\ \E_{i,\g,\h} I_{\mu}\left(X_i;M\mid Y_i, X_\g,Y_\h, W\right) \\
		=&\ \E_{\kappa,s_g,s_h} I_{\mu}\left(X_{\kappa(s_g)};M\mid X_{\kappa([1, s_g-1])},Y_{\kappa([s_h,n])}, W\right) \\
		\intertext{since $s_g$ is uniform between $n/2+1$ and $n$, and by chain rule}
		=&\ \frac{2}{n}\cdot \E_{\kappa,s_h} I_{\mu}\left(X_{\kappa([n/2+1,n])};M\mid X_{\kappa([1,n/2])},Y_{\kappa([s_h,n])}, W\right) \\
		\leq&\  \frac{2}{n}\cdot |M| \\
		=&\ O(C/n).
	\end{align*}

	\paragraph{$\pi_{\ighr}$ is a $\theta$-protocol with small $\theta$.} We use the second view of the distribution. By \eqref{eq_dist_prot}, $\theta$ is at most
	\begin{align*}
		\E_{\ighr} I_{\nu_\ighr}(Y;M\mid X) &\leq \E_\ighr I_{\mu}(Y_i;M\mid X_i,(X_\g,Y_\h)=r, W) \\
		&=\E_{i,\g,\h} I_{\mu}(Y_i;M\mid X_i,X_\g,Y_\h, W) \\
		&=\E_{\kappa,s_g,s_h} I_{\mu}(Y_{\kappa(s_h)};M\mid X_{\kappa([1,s_g])}, Y_{\kappa([s_h+1,n])}, W) \\
		&=\frac{2}{n}\cdot \E_{\kappa,s_g} I_{\mu}(Y_{\kappa([1,n/2])};M\mid X_{\kappa([1,s_g])}, Y_{\kappa([n/2+1,n])}, W).
	\end{align*}

	Note that since $s_g\geq n/2+1$, the mutual information would be $0$ if we \emph{did not} condition on $W$:
	\begin{align*}
		&\ I_{\mu}(Y_{\kappa([1,n/2])};M\mid X_{\kappa([1,s_g])}, Y_{\kappa([n/2+1,n])}) \\
		=&\ H(Y_{\kappa([1,n/2])}\mid X_{\kappa([1,s_g])}, Y_{\kappa([n/2+1,n])}) 
		- H(Y_{\kappa([1,n/2])}\mid M, X_{\kappa([1,s_g])}, Y_{\kappa([n/2+1,n])}) \\
		\leq&\ H(Y_{\kappa([1,n/2])}\mid X_{\kappa([1,n/2])})-H(Y_{\kappa([1,n/2])}\mid M, X^{(n)}, Y_{\kappa([n/2+1,n])}) \\
		=&\ H(Y_{\kappa([1,n/2])}\mid X_{\kappa([1,n/2])})-H(Y_{\kappa([1,n/2])}\mid X, Y_{\kappa([n/2+1,n])}) \\
		=&\ 0.
	\end{align*}
	Therefore, by Lemma~\ref{lem:indep_W_minfo} below, $I_{\mu}(Y_{\kappa([1,n/2])};M\mid X_{\kappa([1,s_g])}, Y_{\kappa([n/2+1,n])}, W)\leq \log\frac{1}{\Pr(W)}$, and hence
	\[
		\E_\ighr I_{\mu}(Y_i;M\mid X_i,(X_\g,Y_\h)=r, W)\leq \frac{2}{n}\cdot\log\frac{1}{\Pr(W)}=O\left(\frac{1}{n}\log\frac1p\right).
	\]

	\paragraph{The expected error probability is low.} By Equation~\eqref{eq_err_prob} and Jensen's inequality, it suffices to upper bound $\E_\ighr[I_{\mu}(X_i;O_i\mid Y_i,M,(X_\g,Y_\h)=r,W)]$.
	We view $(i, \g, \h)$ as a triple sampled from the first distribution.
	\begin{align*}
		&\ \E_\ighr I_{\mu}(X_i;O_i\mid Y_i,M,(X_\g,Y_\h)=r,W) \\
		=&\ \E_{i,\g,\h} I_{\mu}(X_i; O_i\mid Y_i, X_\g, Y_\h, M, W) \\
		\leq &\ \E_{\kappa,s_g,s_h} I_{\mu}(X_{\kappa(s_g)}; O\mid X_{\kappa([1,s_g-1])}, Y_{\kappa([s_h,n])}, M, W) \\
		=&\ \frac{2}{n}\cdot \E_{\kappa,s_h} I_{\mu}(X_{\kappa([n/2+1,n])}; O\mid X_{\kappa([1,n/2])}, Y_{\kappa([s_h,n])}, M, W).
	\end{align*}
	Similarly since $s_h\leq n/2$, $X_{\kappa([n/2+1,n])}$ and $O$ are independent if we did not condition on $W$:
	\begin{align*}
		&\  I_{\mu}(X_{\kappa([n/2+1,n])}; O\mid X_{\kappa([1,n/2])}, Y_{\kappa([s_h,n])}, M) \\
		\leq&\ I_{\mu}(X_{\kappa([n/2+1,n])}; Y_{\kappa([1,s_h-1])}\mid X_{\kappa([1,n/2])}, Y_{\kappa([s_h,n])}, M) \\
		=&\ H(Y_{\kappa([1,s_h-1])}\mid X_{\kappa([1,n/2])}, Y_{\kappa([s_h,n])}, M) 
		- H(Y_{\kappa([1,s_h-1])}\mid X^{(n)}, Y_{\kappa([s_h,n])}, M) \\
		\leq&\ H(Y_{\kappa([1,s_h-1])}\mid X_{\kappa([1,n/2])}, Y_{\kappa([s_h,n])}) 
		- H(Y_{\kappa([1,s_h-1])}\mid X^{(n)}, Y_{\kappa([s_h,n])}) \\
		=&\ 0.
	\end{align*}

	Thus, by Lemma~\ref{lem:indep_W_minfo},
	\[
		\E_\ighr I_{\mu} (X_i;O_i\mid Y_i,M,(X_g,Y_h)=r,W)\leq \frac{2}{n}\cdot \log\frac{1}{\Pr(W)}=O\left(\frac{1}{n}\log\frac1p\right).
	\]

	By Equation~\eqref{eq_err_prob} and Jensen's inequality, the error probability is at most
	\begin{align*}
		&\ O\left(\E_\ighr \sqrt{I_{\mu}(X_i;O\mid Y_i,M,(X_\g,Y_\h)=r,W)}\right) \\
		\leq&\ O\left(\sqrt{\E_\ighr I_{\mu}(X_i;O\mid Y_i,M,(X_\g,Y_\h)=r,W)}\right) \\
		\leq&\ O\left(\sqrt{\frac{1}{n}\log\frac1p}\right).
	\end{align*}

	\paragraph{The expected $\DKL(\cD_\ighr'[X]\|\cD[X])$ is small.} We have 
        \begin{align*}
	\E_\ighr \DKL(\cD_\ighr'[X]\|\cD[X])&= \E_{\stackrel{i,\g,\h}{r\sim \mu[(X_\g, Y_\h)\mid W]}}\DKLver{\mu[X_i\mid (X_\g,Y_\h)=r, W]}{\mu[X_i]} \\
          {}&\leq  \log \frac{1}{\Pr(W)}\\
          {} & \leq \log \frac 1p
	\end{align*}
        since $\DKL(p(X|W) \| p(X)) \le \log(1/\Pr(W))$ for any distribution $p$, random variable $X$, and event $W$.

        \paragraph{The expected $\DKL(\cD_\ighr'[Y]\|\cD[Y])$ is small.}
        Identically argued as the previous case, we also have 
        $$\E_\ighr \DKL(\cD_\ighr'[Y]\|\cD[Y])\le \log \frac 1{\Pr(W)} \le \log \frac 1p .$$

        \paragraph{The expected $\E_{x\sim \cD'[X]}\DKL(\cD'[Y|X=x]\|\cD[Y|X=x])$ is small.}

	We view $(i,\g,\h)$ as being sampled from the second distribution. Then the quantity of interest equals
	\begin{align*}
		&\ \E_{\stackrel{i,\g,\h}{(r,x)\sim \mu[(X_\g, Y_\h, X_i)\mid W]}}\DKLver{\mu[Y_i\mid X_i=x, (X_\g,Y_\h)=r, W]}{\mu[Y_i\mid X_i=x, (X_\g,Y_\h)=r]} \\
		=&\ \E_{\stackrel{\kappa,s_g,s_h}{r'\sim \mu[(X_{\kappa([1,s_g])}, Y_{\kappa([s_h+1,n])})\mid W]}}\DKLver{\mu[Y_{\kappa(s_h)}\mid (X_{\kappa([1,s_g])},Y_{\kappa([s_h+1,n])})=r', W]}{\mu[Y_{\kappa(s_h)}\mid (X_{\kappa([1,s_g])},Y_{\kappa([s_h+1,n])})=r']} \\
		=&\ \frac{2}{n}\cdot \E_{\stackrel{\kappa,s_g}{r'\sim \mu[(X_{\kappa([1,s_g])}, Y_{\kappa([n/2+1,n])})\mid W]}}\DKLver{\mu[Y_{\kappa([1,n/2])}\mid \left(X_{\kappa([1,s_g])},Y_{\kappa([n/2+1,n])}\right)=r', W]}{\mu[Y_{\kappa([1,n/2])}\mid \left(X_{\kappa([1,s_g])},Y_{\kappa([n/2+1,n])}\right)=r']} \\
		\leq&\ \frac{2}{n}\cdot \log\frac{1}{\Pr(W)}.
	\end{align*}

        \paragraph{The expected $\E_{y\sim \cD'[Y]}\DKL(\cD'[X|Y=y]\|\cD[X|Y=y])$ is small.} The argument is identical to the immediately preceding case, with the roles of $x$ and $y$ reversed.

	\bigskip

	Finally, since mutual information, error probability and KL divergence are all non-negative, by Markov's inequality and the union bound, there exists a quadruple $(i, \g, \h, r)$ such that
	\begin{enumerate}
		\item $\pi_{\ighr}$ is an $O(\frac{1}{n}\log\frac1p)$-protocol under $\cD'_{\ighr}$;
                \item $\DKL(\cD'_{i,\g,\h,r}[X]\| \cD[X]), \DKL(\cD'_{i,\g,\h,r}[Y]\| \cD[Y]) \leq O(\log\frac{1}{p})$;
                \item The quantities $\E_{x\sim\cD'_{i,\g,\h,r}[X]} \DKL(\cD'_{i,\g,\h,r}[Y|X=x]\|\cD[Y|X=x])$ and $\E_{y\sim\cD'_{i,\g,\h,r}[Y]} \DKL(\cD'_{i,\g,\h,r}[X|Y=y]\|\cD[X|Y=y])$ are both at most $O(\frac 1n \log\frac 1p)$;
		\item the information cost of $\pi_{i,\g,\h,r}$ on input pair drawn from $\cD'_{i,\g,\h,r}$ is at most $O(C/n)$;
		\item the error probability of $\pi_{i,\g,\h,r}$ on a random instance drawn from $\cD'_{i,\g,\h,r}$ is at most $O(\sqrt{\frac{1}{n}\log\frac{1}{p}})$.
	\end{enumerate}
	This proves the lemma.
\end{proof}

The following appears as \cite[Lemma 19]{BravermanRWY13}.
\begin{lemma}\label{lem:indep_W_minfo}
$A, B$ are independent conditioned on $R$, then for any event $W$, $I(A; B\mid R, W)\leq \log\frac{1}{\Pr(W)}$.
\end{lemma}

\subsection{$n$-fold \URsup{} lower bound}

\begin{lemma}\label{lem:ur}
There is a fixed constant $\delta_0\in(0,1)$ and an input distribution $\cD_\ur$ over pairs $(S,T)$ for \URsup{} such that the following holds. Suppose distribution $\cD'$ satisfies
\begin{itemize}
\item $\DKL(\cD'[S]\|\cD_\ur[S]), \DKL(\cD'[T]\|\cD_\ur[T]) \le O(1)$, where $\cD[X]$ is the marginal distribution of $X$ in $\cD$.
\item Both $\E_{s\sim \cD'[S]}\DKL(\cD'[T|S=s]\|\cD_\ur[T|S=s])$ and  $\E_{t\sim \cD'[T]} \DKL(\cD'[S|T=t]\|\cD_\ur[S|T=t])$ are at most $\eta$, where $\cD_\ur[X|Y=y]$ denotes the marginal distribution of $X$ conditioned on $Y=y$.
\end{itemize}
Then, any one-way $\eta$-protocol $P$ for \URsup{} with error probability $\delta$ over $\cD'$ must have \emph{internal information cost} $\icost$ (i.e., $I(S; M\mid T)$) at least $\icost\geq\Omega(\log \frac{1}{\delta} \log^2 U)$, as long as $\eta\leq O(\delta^2)$ and $2^{-U^{1-\epsilon}}\leq \delta<\delta_0$ for any given constant $\epsilon$. 
\end{lemma}
The proof of Lemma~\ref{lem:ur} is similar to that of Theorem~\ref{thm_ur}~\cite{KapralovNPWWY17}, and is deferred to Appendix~\ref{app_urlb}.
Now we are ready to prove Lemma~\ref{lem_str_lb}, the communication lower bound for $n$-fold \URsup{}.
\begin{proof}[Proof of Lemma~\ref{lem_str_lb}]
Consider any one-way protocol with communication cost $C$ and error probability $\delta$ for $n$-fold \URsup{} on instances sampled from $\cD_\ur^n$.
Then by Lemma~\ref{lem:dp}, there exists a one-way $O(\delta/n)$-protocol $\pi$ and an input distribution $\cD'$ for \URsup{} over $[n]$ such that
\begin{itemize}
        \item $\DKL(\cD'[S]\|\cD_\ur[S]), \DKL(\cD'[T]\|\cD_\ur[T]) \leq O\left(\log\frac{1}{1-\delta}\right)\leq O(\delta) = O(1)$,
        \item Both $\E_{s\sim \cD'[S]}\DKL(\cD'[T|S=s]\|\cD[T|S=s])$ and  $\E_{t\sim cD'[T]} \DKL(\cD'[S|T=t]\|\cD[S|T=t])$ are at most $O\left(\frac 1n\log\frac 1{1-\delta}\right) = O(\delta/n)$, and
	\item when input pair $(S,T)$ is drawn from $\cD'$, $\pi$ has \emph{information} cost $I(S; M\mid T)\leq O(C/n)$ and computes $f$ with probability $1-O\left(\sqrt{\frac{1}{n}\log\frac{1}{1-\delta}}\right)\geq 1-O(\sqrt{\delta/n})$.
\end{itemize}

Finally, as $2^{-n^{1-\epsilon}}\leq \delta/n=o(1)$, by Lemma~\ref{lem:ur} $C/n\geq \Omega(\log \frac n\delta \log^2 n)$.
This proves the lemma.
\end{proof}

\newcommand{\etalchar}[1]{$^{#1}$}

\appendix

\section*{Appendix}

\section{The AGM sketch for small failure probability}\label{sec:agm-small-fail}
The analysis of the AGM sketch in \cite{AhnGM12} shows that dynamic spanning forest can be solved with failure probability $\delta = 1/\mathop{poly}(n)$ using $O(n\log^3 n)$ bits of memory. We remark here that the same algorithm but with a different setting of parameters can achieve arbitrarily small failure probability $\delta\in(0,1)$ using $O(n\log(n/\delta)\log^2 n)$ bits of memory, showing that our lower bound from Theorem~\ref{thm:main-streaming} is optimal for any $\delta > 1/2^{n^{1-\Omega(1)}}$. This modification can also be used to achieve an $O(\log(n/\delta)\log^2 n)$ bit message length per vertex in the distributed model of Section~\ref{sec_dist}. We note that the usual technique to achieve success probability amplification via parallel independent repetition and returning the ``median'' or some such result is not applicable, since a graph may have exponentially many spanning forests and each parallel repetition may output a different one. Thus it would not be clear which spanning forests output across the repetitions are valid, i.e.\ use edges that actually exist in the graph, as all those returned may be distinct, even if correct.

First we recall the support-finding problem variant described in \cite{KapralovNPWWY17}.

\begin{definition}
In the turnstile streaming problem {\em \suppfind{k}$(\delta_1,\delta_2)$}, there is a vector $z\in\R^n$ receiving turnstile streaming updates, and the answer to $\mathbf{query}()$ must behave as follows:
\begin{itemize}
\item With probability at most $\delta_1$, the output can be `\textsf{Fail}'.
\item With probability at most $\delta_2$, the output can be arbitrary.
\item Otherwise, the output should be any subset of size $\min\{k,\|z\|_0\}$ from $\supp(z)$.
\end{itemize}
\end{definition}

We henceforth define $t = \max\{k, \log(1/\delta_1)\}$.

\begin{theorem}[{\cite{KapralovNPWWY17}}]\label{thm:suppfind}
  For any $k\ge 1$ and $0<\delta_1,\delta_2<1$, there is a solution to \suppfind{k}$(\delta_1,\delta_2)$ using $O((t\log n + \log(n/\delta_2))\log(n/t))$ bits of memory. Furthermore the memory contents of this data structure $D$ can be represented by a linear sketch, i.e.\ $\Pi z$ for some matrix $\Pi$.
\end{theorem}

\begin{figure}[h]
\begin{center}
\fbox{                                                                                                         
{\footnotesize                                                                                                 
\parbox{6.375in} {                                                                                             
\underline{$\mathsf{AGM}$ sketch}:\\
\textbf{initialization}
\vspace{-.05in}\begin{enumerate}
\addtolength{\itemsep}{-2mm}
\item $\delta' := \min\{1/(6e), \log(n/\delta)/\log n\}$
\item $R:= \lceil \log_{3/2} n\rceil +  \max\{\lceil\log_{3/2} n\rceil, \log(2/\delta)/\log(1/(6e\delta'))\}$
\item \textbf{for} $u=1,\ldots,n$:
\item[] \qquad \textbf{for} $r=1\ldots R$:
\item[] \qquad\qquad initialize data structure $D_{u,r}$ from Theorem~\ref{thm:suppfind} for \suppfind{1}$(\delta', \frac{\delta}{nR})$ for vector $z_u\in\R^{\binom n2}$ initialized
\item[] \qquad\qquad to $0$, so that $D_{u,r}$ stores $\Pi_r z_u$ in memory.
\end{enumerate} 

\textbf{update$(u,v,\Delta)$} // $\Delta=+1$ signifies adding edge $(u,v)$ to $G$, and $\Delta=-1$ signifies deleting the edge; wlog assume $u<v$
\vspace{-.2in}\begin{enumerate}
\addtolength{\itemsep}{-2mm}
\medskip
\smallskip
\item \textbf{for} $r=1\ldots R$:
\item[] \qquad $D_{u,r}.\mathbf{update}((u,v),+\Delta)$ // i.e.\ process the change $(z_u)_{(u,v)} \leftarrow (z_u)_{(u,v)} + \Delta$
\item[] \qquad $D_{v,r}.\mathbf{update}((u,v),-\Delta)$
\end{enumerate}

\textbf{query}()
\vspace{-.05in}
\begin{enumerate}
\addtolength{\itemsep}{-2mm}
\item $F\leftarrow \emptyset$ // final spanning forest we output
\item $S\leftarrow \{\{1\},\ldots,\{n\}\}$ // current connected components in $F$
\item \textbf{for} $r=1,\ldots,R$:
\item \qquad $A\leftarrow \emptyset$ // edges to be added to $F$ in this iteration
\item \qquad \textbf{for} $s\in S$:
\item[] \qquad\qquad $(u,v)\leftarrow D_{s,r}.\mathbf{query}()$ // $D_{s,r}$ denotes the data structure obtained from summing $\sum_{w\in s}\Pi_{w,r} z_w$
\item[] \qquad\qquad $A\leftarrow A\cup\{(u,v)\}$
\item \qquad $F\leftarrow F\cup A$
\item\qquad \textbf{for} $(u,v)\in A$:
\item[] \qquad\qquad // merge connected components linked by the edge $(u,v)$
\item[] \qquad\qquad identify the sets $s_u,s_v$ containing $u$ (resp.\ $v$) in $S$; remove them each from $S$ and insert their union into $S$.
\item \textbf{return} $F$

\end{enumerate}
}}}
\caption{Dynamic spanning forest algorithm via the AGM sketch. We assume $\delta < 1/n^C$ for some large constant $C$, since otherwise the desired $O(n\log^3 n)$ bits of memory is already achieved in \cite{AhnGM12}.}\label{fig:agm}
\end{center}
\end{figure}

We now give an overview and analysis of the AGM sketch (see Figure~\ref{fig:agm}). We reiterate that the algorithm and analysis presented here are essentially the same as that in the original work \cite{AhnGM12}, though we present all details here to point out what changes need to be made to achieve arbitrarily small failure probability $\delta$. Specifically, the only differences in the algorithm in Figure~\ref{fig:agm} and that in the original work \cite{AhnGM12} which achieved failure probability $1/\mathop{poly}(n)$ is the setting of $\delta'$ in initialization (in \cite{AhnGM12} $\delta'$ was set to $1/10$), which also implies a difference in the value of $R$. We henceforth assume $\delta < 1/\mathop{poly}(n)$ since otherwise the \cite{AhnGM12} analysis already applies.

The sketch's query algorithm to output the spanning forest is iterative, with $R$ rounds. The algorithm explicitly maintains a partition of $[n]$ into connected pieces, initially the partition with $n$ singletons, then in each iteration queries each partition for an edge $e$ leaving that partition (if one exists) to then merge with some other partition which is non-maximally connected. We then add all such edges $e$ found in any given iteration to a forest $F$, which we return at the end of the $R$ rounds. The intent is for these partitions to all be maximal connected components and for $F$ to be a spanning forest by the end of the $R$th round. We find edges to merge non-maximal components as follows. Each vertex $u$ stores $R$ sketches, using independent randomness, of the vector $z_u\in\R^{\binom n2}$ which is the (signed) edge-incidence vector for vertex $u$. That is, if $(u,v)$ is in the graph then $(z_u)_{(u,v)}$ will be $\pm 1$, with the sign determined by whether $u<v$. We let $D_{u,r}$ for $r=1,\ldots,R$ denote these $r$ sketches, each of which solves \suppfind{1}$(\delta',\delta'')$ using the space promised by Theorem~\ref{thm:suppfind}, where $\delta''=\delta/(2nR)$ as seen in Figure~\ref{fig:agm}. Each $D_{u,r}$'s memory contents is $\Pi_r z_u$ for some matrix $\Pi_r$. Then for $A\subset [n]$, we can define $D_{A,r}$ as the data structure whose memory is $ \Pi_r z_A$ with $z_A := \sum_{u\in A} z_u$. The vector $z_A$ has the property that its support is exactly the set of edges leaving $A$ in $G$, so that a correctly answered query to $D_{A,r}$ provides an edge leaving $A$ (if one exists). The space used is 
\begin{equation}
O(nR(\log(1/\delta')\log n + \log(nR/\delta))\log n) . \label{eqn:upper-space}
\end{equation}

We now turn to setting $\delta', R$. Note that if the \suppfind{1} data structures never erred, we could take $R\le\lceil \log_2 n\rceil$ to find a spanning forest since the number of non-maximal components starts off as at most $n$ and at least halves after each round. Now let us take probabilistic errors into account. First, we condition on no $D_{u,r}$ ever outputting a non-existent edge, which happens with probability $1-\delta/2$ by our setting of $\delta''$ and a union bound. Next, call a round ``good'' if at most $k/3$ non-maximal components fail to find an outgoing edge in that round, i.e.\ output `\textsf{Fail}'. Note that in any good round, the number of non-maximal connected components decreases from $k$ to at most $((1-1/3)k)/2 + k/3 = 2k/3$. Thus $F$ is a spanning forest after at most $\lceil \log_{3/2} n\rceil$ good rounds. In any round with $k$ non-maximal components we expect at most $\delta' k$ of them to fail to find an outgoing edge via the \suppfind{1} data structures, so the probability the round is bad is at most $3\delta'$ by Markov's inequality. A simple calculation (see Lemma~\ref{lem:bad-rounds}) then shows that if $R \ge \lceil\log_{3/2} n\rceil + \max\{\lceil\log_{3/2} n\rceil, \Omega(\log(1/\delta)/\log(1/\delta'))\}$, we will have at least $\lceil \log_{3/2} n\rceil$ good rounds with probability $1-\delta/2$, as desired. Substituting for $R$ in \eqref{eqn:upper-space}, our space (in bits) is
\begin{align*}
O(n(\log n + &\log(1/\delta)/\log(1/\delta'))(\log(1/\delta')\log n + \log(n/\delta))\log n)\\
{}& \le O(n(\overbrace{\log n + \log(n/\delta)/\log(1/\delta')}^\alpha)(\overbrace{\log(1/\delta')\log n + \log(n/\delta)}^\beta)\log n) .
\end{align*}
Observe $\beta = \alpha\cdot \log(1/\delta')$. We can thus asymptotically minimize both $\alpha,\beta$ simultaneously by setting $\log(1/\delta') = \Theta(\log(n/\delta))/\log n$, which brings our final space bound to $O(n\log(n/\delta)\log^2 n)$ bits (though for technical reasons, see Lemma~\ref{lem:bad-rounds}, we set $\delta'$ to be the minimum of this quantity and some constant).

\begin{lemma}\label{lem:bad-rounds}
For $R, \delta'$ as in Figure~\ref{fig:agm}, the probability of having less than $\lceil \log_{3/2}n\rceil$ good rounds is at most $\delta/2$.
\end{lemma}
\begin{proof}
As mentioned above, the probability a round is bad is at most $3\delta'$ by Markov's inequality. Thus the probability of not having the desired number of good rounds is at most
\begin{align*}
\binom{R}{R - \lceil \log_{3/2}n\rceil} (3 \delta')^{R - \lceil \log_{3/2}n\rceil} &\le (3e\delta'(1 + \frac{\lceil \log_{3/2} n\rceil}{R - \lceil \log_{3/2} n\rceil}))^{R - \lceil \log_{3/2} n\rceil}\\
{}&\le (6e\delta')^{\log(2/\delta)/\log(1/(6e\delta'))}\\
{}&= \delta/2 .
\end{align*}
\end{proof}

The above yields the following theorem.

\begin{theorem}
The AGM sketch achieves success probability $1-\delta$ using $O(n\log(n/\delta)\log^2 n)$ bits of space.
\end{theorem}

We note that the above sketch can easily be implemented in the distributed sketching model of Section~\ref{sec_dist} by having each vertex $u$ simply send the memory contents of $D_{u,r}$ for $r=1,\ldots,R$ to the referee as a message, who can then run the query algorithm. Thus we also have the following corollary.

\begin{corollary}
In the distributed sketching model with shared public randomness, for any $\delta\in(0,1)$ panning forest can be solved with the maximum message length being at most $O(\log(n/\delta)\log^2 n)$ bits.
\end{corollary}

\section{Proof of Lemma~\ref{lem:ur}}\label{app_urlb}
\paragraph{Hard input distribution $\cD_\ur$.}
Let $m=\sqrt{U\log\frac{1}{\delta}}$ be the size of set $S$, $\alpha=\frac{20}{\log 1/\delta}$.
Let $r_i=\left\lfloor m\cdot (1-(1-\alpha)^i)\right\rfloor$ for $i=0,\ldots,R-1$ be all possible sizes of set $T$, where $R=\left\lfloor\frac1{20\alpha}\log (\alpha m)\right\rfloor$.
In our hard distribution $\cD_\ur$, Alice's input set $S$ is a uniformly random subset of $[U]$ of size $m$.
Then we sample a uniformly random integer $i\in [0, R-1]$, and sample a random subset $T\subseteq S$ of size $r_i$.

\bigskip

To prove the lemma, we are going to use a randomized encoding scheme to encode a random set $S$ of size $m$.
The encoding and decoding procedures will have access to shared random bits, and the encoding procedure has access to additional private random bits.
Then we show that the decoding procedure always reconstructs $S$, but on the other hand, the encoding does not reveal much information about $S$.

Similar to the notations in Section~\ref{sec_str}, let $\cD'[S]$ denote the marginal distribution of $S$ assuming $(S,T)$ is sampled from $\cD'$, for any given $T$ let $\cD'[S\mid T]$ denote the distribution of $S$ conditioned on $T$.

\paragraph{Encoding.} Given a set $S$ drawn from $\cD'[S]$, we use the following encoding procedure to generate $\ENC(S)$.
\begin{enumerate}
	\item Sample a random set $T_0$ from $\cD'[T\mid S]$, and write down $T_0$.
	\item Sample a random message $M$ using the $\eta$-protocol $P$ for input pair $(S, T_0)$, and write down $M$.
	\item Sample a uniformly random permutation $\pi$ over $[U]$ using a \emph{public} random string. 
	\item Set $T=T_0$, $A=\emptyset$, let $i_0$ be the integer such that $|T_0|=r_{i_0}$ (such $i_0$ exists, since $\DKL(\cD'||\cD_\ur)$ is finite and hence $\supp(\cD')\subseteq \supp(\cD_\ur)$).
	\item For $i=i_0,\ldots,R-1$, do the following:
	\begin{enumerate}
		\item Evaluate the output function $O$ of $P$ on input $T$ and message $M$, let $x_i$ be its value;\label{step_a}
		\item If $x_i\in S\setminus T$, $A:=A\cup \{x_i\}$ and $T:=T\cup\{x_i\}$, write down $1$;
		\item Otherwise write down $0$;
		\item Fill $T$ up to $r_{i+1}$ elements (let $r_R=m$) according to $\pi$, i.e., find all elements $x$ in $S\setminus T$, and add the $r_{i+1}-|T|$ with smallest $\pi(x)$ to $T$.
	\end{enumerate}
	\item Write down the set $S\setminus (A\cup T_0)$.
\end{enumerate}

Note that with our setting of parameters, $r_{i+1}-r_i\geq 1$.
Although $M$ is generated for input pair $(S, T_0)$, we still use it on all later inputs $T$ in Step~\ref{step_a}.

\paragraph{Decoding.} The following decoding procedure reconstructs $S$.
\begin{enumerate}
	\item Read set $T_0$, message $M$ and set $S\setminus (A\cup T_0)$.
	\item Set $T=T_0$ and $\overline{A}=S\setminus (A\cup T_0)$, let $i_0$ be the integer such that $|T_0|=r_{i_0}$.
	\item For $i=i_0,\ldots,R-1$, do the following:
	\begin{enumerate}
		\item Evaluate the output function $O$ of $P$ on input $T$ and message $M$, let $x_i$ be its value;
		\item If the next bit is $1$, $T:=T\cup\{x_i\}$;
		\item Fill $T$ up to $r_{i+1}$ elements (let $r_R=m$) according to $\pi$ and $\overline{A}$, i.e., find all elements $x$ in $\overline{A}\setminus T$, and add the $r_{i+1}-|T|$ with smallest $\pi(x)$ to $T$.
	\end{enumerate}
	\item Output $T$.
\end{enumerate}

\paragraph{Analysis.} It is straightforward to verify that for any set $S$, the decoding procedure always successfully reconstructs $S$. 
In the following, we are going to estimate the amount of information $\ENC(S)$ reveals about $S$ (conditioned on the public random bits $\pi$): $I(\ENC(S) ; S\mid \pi)$.
For each step of the encoding procedure, which we reproduce below, we write down its contribution to $I(\ENC(S) ; S\mid \pi)$:
\begin{enumerate}
	\item This step writes down a random set $T_0$ sampled from $\cD'[T\mid S]$: $I_{\cD'}(S;T)$;
	\item This step writes down the message $M$. Since all $M$, $S$ and $T_0$ are independent of the public random bits $\pi$, this step contributes $I(M; S\mid T_0, \pi)=\icost$;
	\item No bit is written in this step;
	\item No bit is written in this step;
	\item Exactly $R$ bits are written, contributing at most $R$;
	\item The entropy of the bits written in this step is at most $H(|A|)+\E_{A}[\log \binom{U-|T_0|}{m-|T_0|-|A|}]$, where $|A|\leq R$.
\end{enumerate}

The following claim asserts that $I_{\cD'}(S;T)$ is small.
\begin{claim}\label{cl_ist}
$I_{\cD'}(S;T)\leq \eta+\log \binom{U}{m}-\E_{|T|\sim\cD'[|T|]}\log\binom{U-|T|}{m-|T|}$.
\end{claim}

We defer the proof of the claim to the end of section. 
Since $S$ can be reconstructed from $\ENC(S)$ and $\pi$, we have $H(S\mid \ENC(S), \pi)=0$.
Assuming Claim~\ref{cl_ist}, we have
\begin{align*}
	H(S)&=I(\ENC(S); S\mid \pi) \\
	&=I(T_0;S\mid \pi)+I(M;S\mid T_0,\pi)+I(\ENC(S) \setminus (T_0,M);S\mid (T_0,M),\pi) \\
	&\leq I_{\cD'}(S;T_0)+I_{\cD'}(M;S\mid T)+H(\ENC(S)\setminus (T_0,M)) \\
	&\leq \eta+\log \binom{U}{m}-\E_{|T_0|}\log\binom{U-|T_0|}{m-|T_0|} +\icost+R+\overbrace{H(|A|)}^{\le \log R} + \E_{|T_0|,|A|}\log \binom{U-|T_0|}{m-|T_0|-|A|}\\
	&\leq \icost+O(R)+\log \binom{U}{m}-\E_{|A|,|T_0|} \left(\log \binom{U-|T_0|}{m-|T_0|}-\log \binom{U-|T_0|}{m-|T_0|-|A|}\right)
\end{align*}

We also have
\begin{align*}
	\log \binom{U-|T_0|}{m-|T_0|}-\log \binom{U-|T_0|}{m-|T_0|-|A|}&=\log \frac{(m-|T_0|-|A|)!(U-m+|A|)!}{(m-|T_0|)!(U-m)!} \\
	&= \log \frac{(U-m+1)(U-m+2)\cdots(U-m+|A|)}{(m-|T_0|)(m-|T_0|-1)\cdots (m-|T_0|-|A|+1)} \\
	&\geq \log \frac{(U-m)^{|A|}}{m^{|A|}} \\
	&=|A|\log \frac{U-m}{m}.
\end{align*}

Therefore, $H(S)\leq \icost+O(R)+\log \binom{U}{m}-(\E |A|)\cdot \log \frac{U-m}{m}$.

On the other hand, since $S$ is drawn from $\cD'[S]$, $\DKL(\cD'[S]\|\cD_\ur[S])\leq O(1)$ and $\cD_\ur[S]$ is uniform, $H(S)\geq \log\binom{U}{m}-O(1)$.
Thus, $\icost\geq (\E|A|)\cdot \log\frac{U-m}{m}-O(R)$. It suffices to lower bound $\E|A|$.
By linearity of expectation, $\E(|A|\mid i_0)=\sum_{i=i_0}^{R-1} \Pr(x_i\in A)$, where $x_i$ is the output in the $i$-th encoding round.

Before we wrap up our analysis, we state a few claims and provide some intuition. Recall that in Step~\ref{step_a}, the encoder uses $M$ for input pair $(S, T_0)$ to compute the function value on other $T$. Observe that since $M$ does not depend too much on $T_0$ and $\cD'$ is close to $\cD_\ur$, the error probability on a random input $T$ should also be small, which we state in the following claim.
\begin{claim}\label{cl_err_prob}
Given a random $S,T_0,M$ according to the input distribution $\cD'$ and the encoding procedure, for a random $T$ sampled from $\cD_\ur[T\mid S]$, $x=O(M, T)$ is in $S\setminus T$ with probability at least $1-O(\delta)$.
\end{claim}

We have $x_i\in A$ if and only if protocol $P$ outputs a correct answer on the set $T$ of round $i$, which has size $r_i$.
If $T$ was a uniformly random subset of $S$ of size $r_i$, this probability would be $\Pr(\textrm{$P$ is correct}\mid |T|=r_i)$.
However, $T$ might not be uniform, since it depends on the outputs of the previous rounds.

The process of generating $T$ for round $i$ can be viewed as follows: if $x_{i_0}$ is a correct output, add a uniformly random subset of $S$ of size $r_{i_0+1}-r_{i_0}$ which contains $x_{i_0}$ to $T$ (i.e., a uniformly random subset conditioned on it containing $x_{i_0}$), otherwise, add a uniformly random subset of $S$ of size $r_{i_0+1}-r_{i_0}$ to $T$; if $x_{i_0+1}$ is a correct output, add a random subset of size $r_{i_0+2}-r_{i_0+1}$ containing $x_{i_0+1}$ to $T$, otherwise, add a random subset of size $r_{i_0+2}-r_{i_0+1}$; and so on.
If all rounds were adding uniformly random subsets to $T$, $T$ would have been uniform. 
But in a subset of the rounds, we are adding random subsets containing certain elements.
We claim that $T$ is actually not far from uniform. The proof of the following claim is similar to that of \cite[Lemma 5]{KapralovNPWWY17}.
\begin{claim}\label{cl_t0}
	Suppose $T_0$ is sampled according to $\cD_\ur[T\mid S]$. Then for any $S$, in round $i\geq i_0$, for any $T'$, $\Pr(T=T'\mid i_0)\leq \frac{1}{\binom{m}{r_i}}\cdot \delta^{-0.9}$.
\end{claim}
That is, the probability of each singleton event may only increase by a factor of $\delta^{-0.9}$.

Fix an $s\in\supp(\cD'[S])$, we first analyze $\E (|A|\mid S=s)$ \emph{assuming} that $T_0$ were drawn according to $\cD_\ur$. Then, assuming the claim, we have for any $i\geq i_0$ that
\[
	\Pr(x_i\notin A\mid i_0, S=s)\leq \Pr_{T\sim \cD_\ur[T\mid S=s]}(P \textrm{ is \emph{incorrect} on } (s,T)\text{ when using } M\mid |T|=r_i,S=s)\cdot \delta^{-0.9}.
\]
Therefore, we have
\begin{align*}
\E_{T_0\sim \cD_\ur[T\mid S=s]}(|A|\mid S) &= \sum_{i_0=0}^{R-1} \Pr(i_0)\cdot \E (|A| \mid i_0,S=s)\\
&{}\geq \sum_{i_0=0}^{R-1} \Pr(i_0) \cdot (R-i_0-\sum_{i=i_0}^{R-1}\Pr_{T\sim \cD_\ur[T\mid S=s]}(P\textrm{ is \emph{incorrect} using $M$}\mid |T|=r_i,S=s)\cdot \delta^{-0.9})\\
&{}\geq \sum_{i_0=0}^{R-1} \Pr(i_0) \cdot (R - i_0 - \Pr_{T\sim \cD_\ur[T\mid S=s]}(P\textrm{ is \emph{incorrect} using $M$}\mid S=s)\cdot R\cdot \delta^{-0.9})\\
&{} = R/2 - \Pr_{T\sim \cD_\ur[T\mid S=s]}(P\textrm{ is \emph{incorrect} on $(s,T)$ using $M$})\cdot R\cdot \delta^{-0.9}.
\end{align*}

Now, note that for a random $S\sim \cD'[S]$, we have
\[
	\E_S\DKL(\cD'[T|S]\|\cD_\ur[T|S])\leq \eta
\]
by the premise of the lemma, and
\[
	\E_S \Pr_{T\sim \cD_\ur[T\mid S]}(P\textrm{ is \emph{incorrect} on $(S,T)$ using $M$})\leq O(\delta)
\]
by Claim~\ref{cl_err_prob}.
Then by Markov's inequality and union bound, for a constant fraction of $s$ according to $\cD'[S]$, we have both
\[
	\DKL(\cD'[T|S=s]\|\cD_\ur[T|S=s])\leq O(\eta)
\]
and
\[
	\Pr_{T\sim \cD_\ur[T\mid S=s]}(P\textrm{ is \emph{incorrect} on $(s,T)$ using $M$})\leq O(\delta).
\]

For any such $s$, by Pinsker's inequality, $|\cD'[T\mid S=s]-\cD_\ur[T\mid S=s]|\leq O(\sqrt{\eta})$. Note also $|A|\in[0, R]$ always. Thus for such $s$,
\begin{align*}
	\E_{T_0\sim\cD'[T\mid S=s]} (|A|\mid S=s)&\geq \E_{T_0\sim\cD_\ur[T\mid S=s]} (|A|\mid S=s) - O(R\cdot \sqrt{\eta}) \\
	&\geq R/2-\Pr_{T\sim \cD_\ur[T\mid S=s]}(P\textrm{ is \emph{incorrect} on $(s,T)$ using $M$})\cdot R\cdot \delta^{-0.9}-O(R\cdot \sqrt{\eta}) \\
	&\geq R/2-O(R\cdot \delta^{0.1}+R\cdot\sqrt{\eta}).
\end{align*}

As long as $2^{-U^{1-\epsilon}}\leq \delta<\delta_0$, $\E(|A|\mid S=s)\geq \Omega(R)$.
Since $|A|$ is always nonnegative, and such $s$ occurs with a constant probability, we have $\E |A|\geq \Omega(R)$.
Therefore, the information cost $\icost$ is at least $\Omega\left(R\cdot \log \frac{U-m}{m}\right)=\Omega(\log\frac{1}{\delta}\log^2 U)$.

Now we prove the three claims.

\begin{proof}[Proof of Claim~\ref{cl_ist}]
By the premise of the lemma, we have
\[
	\E_{t\sim \cD'[T]}\DKL(\cD'[S\mid T=t]\|\cD_{\ur}[S\mid T=t])\leq \eta.
\]
$\cD_\ur[S\mid T=t]$ is a uniform distribution with entropy $\log\binom{U-|t|}{m-|t|}$.
Hence, we have
\begin{align*}
	H_{\cD'}(S\mid T)&=\E_{t\sim\cD'[T]} H_{\cD'}(S\mid T=t) \\
	&\geq \E_{|T|\sim \cD'[|T|]}\log\binom{U-|T|}{m-|T|}-\eta,
\end{align*}
where we use the fact that $\DKL(q\|p)=H(p)-H(q)$ for uniform $p$.

Finally, by definition,
\begin{align*}
	I_{\cD'}(S;T)&=H_{\cD'}(S)-H_{\cD'}(S\mid T) \\
	&\leq \log \binom{U}{m}-\E_{|T|\sim \cD'[|T|]}\log\binom{U-|T|}{m-|T|}+\eta.
\end{align*}
This proves the claim.
\end{proof}

\begin{proof}[Proof of Claim~\ref{cl_err_prob}]
Since $P$ is a $\eta$-protocol, by definition, $I(M;T_0\mid S)\leq \eta$.
By Pinsker's inequality, and the fact that $I_p(X;Y\mid Z)=\E_{y,z}\DKL(p[X\mid Y=y,Z=z]\| p[X\mid Z=z])$,
\begin{align*}
	I(M;T_0\mid S)&=\E_{t_0,s}\DKLver{\cD'[M\mid T_0=t_0,S=s]}{\cD'[M\mid S=s]} \\
	&\geq \Omega(\E_{t_0,s}|\cD'[M\mid T_0=t_0,S=s]-\cD'[M\mid S=s]|^2) \\
	&\geq \Omega\left(\left(\E_{t_0,s}|\cD'[M\mid T_0=t_0,S=s]-\cD'[M\mid S=s]|\right)^2\right).
\end{align*}
The expected statistical distance $\E|\cD'[M\mid T_0,S]-\cD'[M\mid S]|\leq O(\sqrt{\eta})\leq O(\delta)$.
By triangle inequality, for $T\sim \cD'[T\mid S]$, $\E|\cD'[M\mid T_0,S]-\cD'[M\mid T,S]|\leq O(\delta)$.
Since the error probability of $P$ is at most $\delta$, it implies that even if we use $M$ generated from another set $T_0$, since the statistical distance is small, the error probability is at most $O(\delta)$.

Finally, for $T\sim \cD_\ur[T\mid S]$, since $\E|\cD_\ur[T\mid S]-\cD'[T\mid S]|\leq O(\delta)$, the error probability for $(S, T)$ is also upper bounded by $O(\delta)$.
\end{proof}

\begin{proof}[Proof of Claim~\ref{cl_t0}]
Let us upper bound the probability that $T=T'$ for any $T'$ in round $i$:
\allowdisplaybreaks
\begin{align*}
	\Pr(T=T'\mid i_0) &\leq \frac{\binom{r_i}{r_{i_0}}}{\binom{m}{r_{i_0}}}\cdot \prod_{j=i_0}^{i-1}\frac{\binom{r_i-r_j-1}{r_{j+1}-r_j-1}}{\binom{m-r_j-1}{r_{j+1}-r_j-1}} \\
	&= \frac{r_i!\cdot (m-r_{i_0})!}{m!\cdot (r_i-r_{i_0})!}\cdot\prod_{j=i_0}^{i-1}\frac{(r_i-r_j-1)!(m-r_{j+1})!}{(r_i-r_{j+1})!(m-r_j-1)!} \\
	&= \frac{r_i!(m-r_i)!}{(r_i-r_i)!m!}\prod_{j=i_0}^{i-1}\frac{(r_i-r_j-1)!(m-r_j)!}{(r_i-r_j)!(m-r_j-1)!} \\
	&= \frac{1}{\binom{m}{r_i}}\prod_{j=i_0}^{i-1}\frac{m-r_j}{r_i-r_j} \\
	&\leq \frac{1}{\binom{m}{r_i}}\prod_{j=0}^{i-1}\frac{m-m(1-(1-\alpha)^j)+1}{m(1-(1-\alpha)^j)-m(1-(1-\alpha)^i)-1} \\
	&\leq \frac{1+o(1)}{\binom{m}{r_i}}\prod_{j=0}^{i-1}\frac{(1-\alpha)^j}{(1-\alpha)^j-(1-\alpha)^i} \\
	&= \frac{1+o(1)}{\binom{m}{r_i}}\prod_{j=1}^{i}\frac{1}{1-(1-\alpha)^j}.
\end{align*}

Note that $(1-\alpha)^j\leq 1-\frac{1}{3}\alpha j$ as long as $1\leq j\leq 1/\alpha$.
We have
\[
	\prod_{1\leq j\leq 1/\alpha}\frac{1}{1-(1-\alpha)^j}\leq \prod_{1\leq j\leq 1/\alpha}\frac{3}{\alpha j}\leq (3/\alpha)^{1/\alpha}\cdot (e\alpha)^{1/\alpha}=(3e)^{1/\alpha}.
\]
On the other hand, when $j>1/\alpha$, $\frac{1}{1-(1-\alpha)^j}\leq e^{2(1-\alpha)^j}$.
Hence,
\[
	\prod_{j>1/\alpha}\frac{1}{1-(1-\alpha)^j}\leq e^{2\sum_{j>1/\alpha}(1-\alpha)^j}\leq e^{2/(e\alpha)}.
\]
Therefore, we have $\Pr(T=T')\leq \frac{1+o(1)}{\binom{m}{r_i}}\cdot (3e\cdot e^{2/e})^{1/\alpha}\leq \frac{1}{\binom{m}{r_i}}\cdot \delta^{-0.9}$.
\end{proof}

\end{document}